\documentclass{amsart}

\usepackage{amsmath,amssymb}
\usepackage{amsthm}
\usepackage{amsrefs}
\usepackage{indentfirst}
\usepackage{mathrsfs}
\usepackage{hyperref}
\usepackage{xcolor}
\usepackage[margin=1.4in]{geometry}
\usepackage{nicefrac}
\usepackage{xifthen}
\usepackage{physics}
\usepackage{algorithm}
\usepackage{algorithmicx}
\usepackage{algpseudocode}
\usepackage{caption}
\usepackage{subcaption}
\usepackage{graphics}
\usepackage{tikz}
\usetikzlibrary{positioning}
\usepackage{multirow}
\usepackage{booktabs}
\usepackage[justification=centering]{caption}
\usepackage{threeparttable}
\usepackage{bbding}
\usepackage{pifont}

\usepackage[font=small,labelfont=bf,
   justification=justified,
   format=plain]{caption}

\newtheorem{theorem}{Theorem}[section]
\newtheorem{lemma}{Lemma}[section]

\renewcommand{\norm}[1]{\left\|#1\right\|}
\newcommand{\fnorm}[1]{\left\|#1\right\|_{\mathrm{F}}}

\newcommand{\bracket}[1]{\left(#1\right)}
\renewcommand{\abs}[1]{\left|#1\right|}
\renewcommand{\tr}[1]{\mathrm{tr}\left(#1\right)}

\newcommand{\bbC}{\mathbb{C}}
\newcommand{\bbE}{\mathbb{E}}
\newcommand{\bbP}{\mathbb{P}}

\newcommand{\calE}{\mathcal{E}}
\newcommand{\calU}{\mathcal{U}}

\newcommand{\Ksq}[1]{K_{\mathrm{sq}}\left(#1 \right)}
\newcommand{\Kdq}[1]{K_{\mathrm{dq}}\left(#1 \right)}

\title[Analysis of Error Propagation in Quantum Computers]{Analysis of Error Propagation in Quantum Computers}
\author{Ziang Yu}
\address{(ZY) School of Mathematical Sciences, Shanghai Jiao Tong
University, Shanghai, China. }
\email{yuziang@sjtu.edu.cn}
\author{Yingzhou Li}
\address{(YL) School of Mathematical Sciences, Fudan University,
Shanghai, China. }
\email{yingzhouli@fudan.edu.cn}

\begin{document}

\bibliographystyle{plain}

\begin{abstract}

Most quantum gate errors can be characterized by two error models, namely
the probabilistic error model and the Kraus error model. We proved that
for a quantum circuit with either of those two models or a mix of both,
the propagation error in terms of Frobenius norm is upper bounded by $2(1
- (1 - r)^m)$, where $0 \le r < 1$ is a constant independent of the qubit
number and circuit depth, and $m$ is the number of gates in the circuit.
Numerical experiments of synthetic quantum circuits and quantum Fourier
transform circuits are performed on the simulator of the IBM Vigo quantum
computer to verify our analytical results, which show that our upper bound
is tight.

\end{abstract}

\maketitle

\section{Introduction}

Quantum computing has been developing rapidly in recent years. It has been
shown that for some specific tasks, quantum algorithms are faster than
their classical counterparts, including Deutsch-Jozsa
algorithm~\cite{Deutsch1992}, Simon algorithm~\cite{Simon1997}, Schor
algorithm~\cite{Schor1994}, Grover algorithm~\cite{Grover1996}, etc.
However, limited by the current quantum hardware technology, quantum
computers suffer from loads of noise and errors, e.g., depolarization,
decoherence, readout error, etc. For quantum circuits with large depths,
the results are not reliable. The current status of quantum computing is
known as the noisy intermediate-scale quantum (NISQ)
era~\cite{Gallot2017}, which could last for many more years.

There are various errors in executing a quantum algorithm on a
quantum computer. We group them into three categories: quantum algorithm
approximation error, quantum sampling error, and quantum machine error.
Quantum algorithm approximation error is due to the approximation in
representing the original models or problems in the algorithm design. One
typical example is the Trotter error in the quantum phase estimation
algorithm. Quantum sampling error is due to the population mean in
approximating the underlying wavefunction coefficients. Quantum machine
error is due to imperfect hardware, where the major source is caused by
the interaction of the quantum computer with its surrounding environment.
Throughout this paper, we refer to the quantum machine error as the
quantum error and discuss its propagation behavior.

Numerous methods have been proposed to mitigate quantum error. We
group these methods into two categories: quantum error correction and
quantum algorithm design. Quantum error correction adopts quantum syndrome
measurement to provide information about whether and in what ways a qubit
has been corrupted without destroying the quantum state of this logical
qubit. Different quantum error correction codes have been brought forward,
including Shor code~\cite{Shor1995}, Calderbank-Shor-Steane (CSS)
code~\cite{Calderbank1996, Steane1996}, additive
codes~\cite{Calderbank1997, Calderbank1998, Gottesman1996}, etc. From a
quantum algorithm design perspective, the noisy terms could be summed
together, and by the central limit theorem, the summed error would be
mitigated. For example, variational quantum eigensolver is found to be
relatively robust to quantum noises~\cite{Sharma2020, Zeng2021}. A similar
phenomenon is observed in its closely related excited state
eigensolver~\cite{Bierman2022}. For all aforementioned methods mitigating
quantum errors, none of them eliminates the errors. The propagation error
of a noisy quantum circuit guides experiments on how large a quantum
circuit is permitted given a fixed error level. For quantum error
correction, the propagation error could indicate which type of error has a
stronger impact on the final results. Hence, it is essential to study the
cumulation and propagation of the quantum error and give a theoretical
upper bound.

The propagation of quantum error has been studied under various scenarios.
In \cite{Muller-Hermes2016}, the convergence of continuous-time
depolarizing channels was investigated in terms of relative entropy.
Deshpande et al.~\cite{Deshpande2021} gave tight bounds on the convergence
of noisy random circuits. From one perspective, noisy random circuits can
be viewed as the propagation of a sequence of noisy identity gates.
Flannigan et al.~\cite{Flannigan2022} numerically explored the propagation
of quantum errors in simulating the Hubbard model and transverse field
Ising model. Very differently, in this work, we study the propagation of
quantum error with the Kraus and probabilistic error model.

Any quantum algorithm is first compiled into a sequence of quantum
circuits and then executed on a quantum computer or quantum simulator. The
execution of a quantum circuit is equivalent to applying a sequence of
quantum gates on an initial density matrix. Quantum gates are unitary
matrices, and applying a gate on a density matrix admits $U \rho
U^\dagger$, where $\rho$ denotes the density matrix and $U$ is the unitary
matrix associated with the quantum gate. The application of a noisy
quantum gate admits $U\rho U^\dagger + \calE$, where $\calE$ denotes the
error. For various types of quantum errors, $\calE$ admits different
properties. When a trace-preserving error is considered, $\calE$ has trace
zero, i.e., $\tr{\calE} = 0$. Alternatively, we could also represent
quantum error as $\calE \bracket{U \rho U^\dagger}$, where $\calE(\cdot)$
is abused as a linear error operator. For more details on the quantum
error model, please refer to Section~\ref{sec:preliminary}.

In traditional numerical analysis, e.g., numerical ordinary differential
equation error analysis, the global error in terms of matrix norm grows
exponentially with the number of matrices, where the matrix spectrum is
assumed to be greater than one. In quantum computing, all matrices are
unitary with spectrums precisely being one, and we could easily give an
error bound growing linearly instead. While linear growing bound is not
consistent with numerical experiments.

\begin{figure}[htb]
    \centering
    \includegraphics[width=0.6\textwidth]{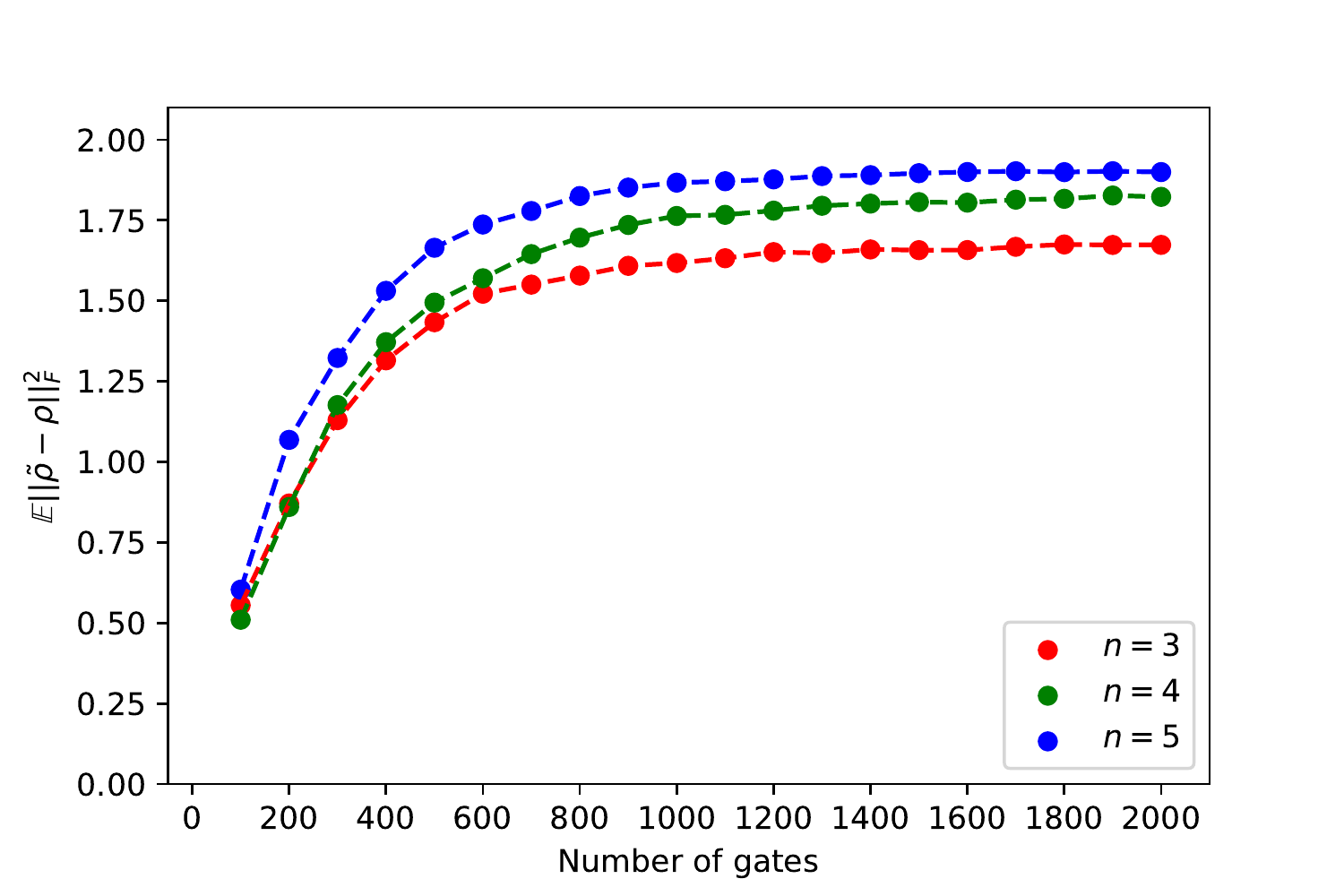}
    \caption{Propagated error for Quantum Fourier Transform (QFT). $3$-,
    $4$-, and $5$-qubit QFT circuits are simulated repeatedly using the
    FakeVigo backend provided by IBM Quantum Experience. The propagated
    error of the circuit is described by
    $\mathbb{E}\fnorm{\tilde{\rho}-\rho}^2$, where $\rho$ represents the
    ideal resulting density matrix, and $\tilde{\rho}$ represents the
    actual resulting density matrix of the noisy circuit. Here the
    expectation is estimated by the empirical mean of 8192 executions of
    the truncated circuits.}\label{fig:QFT-dem}
\end{figure}

Figure~\ref{fig:QFT-dem} illustrates the quantum error of a Quantum
Fourier Transform circuit on a quantum simulator. Instead of growing
linearly, the error saturated and hit a plateau towards the end. Based on
this observation, we aim to give a bound on quantum error propagation,
revealing such a growing behavior. 

In this paper, we analyzed the two error models: the probabilistic error
model and Kraus error model, and proved their bounds of error propagation
in terms of matrix Frobenius norm. Both bounds characterize the error
growth behavior as in Figure~\ref{fig:QFT-dem}. Then, we combined two
results together and proved the following main theorem for the mixed error
model. 
\begin{theorem}
    \label{thm:main_thm}
    Given a quantum circuit with a sequence of single-qubit or
    double-qubit gates $G_1, \cdots, G_m$ starting with an initial density
    matrix $\rho_0$. Each gate $G_k$ is implemented with a probabilistic
    error $P_k$ with error probability $p_k \in [0, 1)$, or a Kraus error
    in the form of \eqref{eq:Kraus_sq} or \eqref{eq:Kraus_dq}. Then the
    expected error propagation is bounded as,
    \begin{equation*}
        \bbE \fnorm{\tilde{\rho}_m - \rho_m}^2 \le 2(1 - (1 - r)^m),
    \end{equation*}
    where $r$ is a constant, $0 \le r < 1$, independent of the number of
    qubits.
\end{theorem}

In addition to the theoretical bound, we verify our analysis bound using
numerical experiments on quantum simulators, including a simulator of IBM
Vigo quantum computer. Numerical results for the Kraus and
probabilistic models are performed on quantum simulators to demonstrate the
tightness of our bounds. Also, numerical examples of QFT with various
numbers of qubits are included to verify our analytical results.

The rest paper is organized as follows. In Section~\ref{sec:preliminary},
two types of error models are included and explained in detail. A formal
statement of the linearly growing quantum error bound is given. The
refined analysis of the error bounds for Kraus and probabilistic models
are given in Section~\ref{sec:kraus-bound} and
Section~\ref{sec:prob-mix-bound}, respectively. In
Section~\ref{sec:prob-mix-bound}, we also give a bound on the mixed error
model, i.e., prove Theorem~\ref{thm:main_thm}. Section~\ref{sec:numerical}
shows the numerical experiments verifying our bounds and demonstrates the
tightness. Finally, we conclude our paper in Section~\ref{sec:conclusion}
with a discussion on future work.

\section{Quantum Error Models and Linear Error Propagation}
\label{sec:preliminary}

The basic unit of quantum computing is quantum bits, or qubits, which is
the quantum counterpart of bits in classical computers. A qubit is a
two-state quantum-mechanical system. Different from a classical bit, a
qubit can be in a coherent superposition of both states simultaneously.
Mathematically, the possible states of an $n$-qubit system form an
$N$-dimensional Hilbert space, where $N = 2^n$, and therefore can be
described by an $N$-dimensional complex vector. Such a state is called a
pure state. 

When quantum errors are included, the quantum system is no longer isolated
from the environment and interacts with the surrounding environment. Then
the system cannot be described as a pure state. Instead, it can be
described as a probabilistic mixture of a set of pure states. Therefore,
density matrices should be used to describe such a mixed state. For an
$n$-qubit quantum system, the state of the system can be described as an
$N \times N$ density matrix $\rho$. A density matrix is a semi-positive
definite Hermitian matrix with trace being $1$. A useful property is that
$\fnorm{\rho} \le 1$, and $\fnorm{\rho} = 1$ if and only if $\rho$
represents a pure state. 

A quantum algorithm in quantum computing is modeled and compiled into a
quantum circuit, where the quantum circuit is composed of a sequence of
quantum gates. Basic single-qubit quantum gates include Pauli gates
($X,Y,Z$), Hadamard gate ($H$), phase gate ($S$), etc. Double-qubit gates
include controlled not gate ($CNOT$), controlled $Z$ ($CZ$), etc. A
single-qubit and a double-qubit gate can be described as a two-dimensional
and a four-dimensional unitary matrix, respectively. For an $n$-qubit
quantum circuit, a single-qubit operator $U$ acting on the $j$-th qubit
can be described as an $N$-dimensional unitary matrix admitting a tensor
product form $I \otimes \cdots \otimes I \otimes U \otimes I \otimes
\cdots \otimes I$, where $I$ is a two-dimensional identity matrix and $U$
appears at the $j$-th position. Double-qubit operators on $n$-qubit circuit
admit a similar tensor product form with two positions replaced by the
4-dimensional submatrix. Therefore, each gate acting on an $n$-qubit
quantum circuit can be described as a unitary matrix $U\in\mathcal{U}(N)$,
where $\calU(N)$ denotes the set of all unitary matrices of size $N$ by
$N$. Therefore, applying a quantum gate $G$ on a quantum state $\rho$
leads to a new state $U \rho U^\dagger$, where $U$ is the underlying
unitary matrix of $G$. 

\subsection{Quantum error models}

There are two widely adopted mathematical models describing quantum
errors, namely the probabilistic error model and the Kraus error model.
Under the probabilistic error model, after a gate $G$ is applied on some
qubits, there is a nonzero probability $p > 0$ that another error operator
is applied on the same qubits. The error operator could be $X, Y, Z$,
reset, or other operators. We define the probabilistic error operator as
\begin{equation}\label{eq:Probability}
    P(\rho) =
    \begin{cases}
    \rho, & \text{with probability } 1-p \\
    \text{another state}, & \text{with probability } p
    \end{cases}.
\end{equation}
Throughout this paper, the probabilistic error model is used with a set of
error operators, i.e., bit flip ($X$ error), phase flip ($Z$ error),
bit-phase flip ($Y$ error), reset error, and depolarizing error.

Under the Kraus error model, after a gate $G$ is applied, a Kraus operator
will be applied afterward, where the Kraus operator admits,
\begin{equation*}
    K(\rho) = V_1 \rho V_1^\dagger + \cdots + V_K \rho V_K^\dagger,
\end{equation*}
for
\begin{equation}\label{eq:normalize}
    \sum_{k=1}^K V_k^\dagger V_k = I.
\end{equation}
Here $\{V_k\}_{k=1}^K$ are gate $G$ dependent. Two Kraus error examples
are amplitude damping and phase damping. Both errors work as
$\mathcal{E}(\rho)=V_1\rho V_1^\dagger + V_2\rho V_2^\dagger$, with
\begin{equation*}
    V_1 =
    \begin{pmatrix}
        1 & 0 \\
        0 & \sqrt{1-\gamma}
    \end{pmatrix},
    V_2 =
    \begin{pmatrix}
        0 & \sqrt{\gamma} \\
        0 & 0
    \end{pmatrix}, 
\end{equation*}
for amplitude damping, and
\begin{equation*}
    V_1 =
    \begin{pmatrix}
        1 & 0 \\
        0 & \sqrt{1-\lambda}
    \end{pmatrix}, 
    V_2 =
    \begin{pmatrix}
        0 & 0 \\
        0 & \sqrt{\lambda}
    \end{pmatrix}, 
\end{equation*}
for phase damping, where $\gamma$ and $\lambda$ are parameters in
amplitude damping and phase damping, respectively. 

The probabilistic error model and Kraus error model appear in quite
different forms, but they are deeply related. In fact, the probabilistic
error model can be equivalently written in the Kraus format. A few typical
examples are included in Appendix~\ref{app:kraus-prob}.

As we mentioned earlier, quantum gates are applied to one or two qubits.
Both the probabilistic error model and Kraus error model we discuss in
this work are associated with quantum gates and are applied to the same
qubits after the gate operation. The probabilistic error models on one or
two qubits are the same as \eqref{eq:Probability}. For Kraus error models,
we focus on specific forms for single-qubit and double-qubit systems,
which are widely adopted in quantum simulators and cover a wide range
of quantum errors. The Kraus error model for single-qubit systems admits,
\begin{equation} \label{eq:Kraus_sq}
    \Ksq{\rho} = V_1\rho V_1^{\dagger} + V_2\rho V_2^{\dagger}
    + V_3\rho V_3^{\dagger} + V_4\rho V_4^{\dagger},
\end{equation}
where
\begin{equation} \label{eq:Kraus_sq_matrix_form}
    V_1 =
    \begin{pmatrix}
        a_1 & 0 \\
        0 & b_1
    \end{pmatrix},
    \quad
    V_2 =
    \begin{pmatrix}
        a_2 & 0 \\
        0 & b_2
    \end{pmatrix},
    \quad
    V_3 = 
    \begin{pmatrix}
    0 & 0 \\
    a_3 & 0
    \end{pmatrix},
    \quad
    V_4 =
    \begin{pmatrix}
    0 & b_3 \\
    0 & 0
    \end{pmatrix},
\end{equation}
satisfying $V_1^\dagger V_1+V_2^\dagger V_2+V_3^\dagger V_3+V_4^\dagger
V_4=I$. The equality constraint on $V_i$s is equivalent to
\begin{equation*}
    a_1^2 + a_2^2 + a_3^2 = 1 \quad \text{ and } \quad
    b_1^2 + b_2^2 + b_3^2 = 1.
\end{equation*}
Without loss of generality, we assume that $\det{\begin{pmatrix} a_1 & a_2
\\ b_1 & b_2 \end{pmatrix}} \neq 0$. 

For a double-qubit system, if both qubits have Kraus errors in the form of
\eqref{eq:Kraus_sq} being
\begin{equation*}
    \begin{split}
        K_1(\rho) & = V_{11}\rho V_{11}^\dagger + V_{12}\rho V_{12}^\dagger
        + V_{13}\rho V_{13}^\dagger + V_{14}\rho V_{14}^\dagger, \\
        K_2(\rho) & = V_{21}\rho V_{21}^\dagger + V_{22}\rho V_{22}^\dagger
        + V_{23}\rho V_{23}^\dagger + V_{24}\rho V_{24}^\dagger, 
    \end{split}
\end{equation*}
then the effect on the double-qubit system can be written as a Kraus
model
\begin{equation} \label{eq:Kraus_dq}
    \Kdq{\rho} = (K_1\otimes K_2)(\rho)
    = \sum_{j=1}^{16} V_j \rho V_j^{\dagger},
\end{equation}
where $\rho\in\mathbb{C}^{4\times 4}$ represents the density matrix of the
double-qubit systems, and matrices are tensor products $\{V_j\}_{j=1}^{16}
= \{V_{1i} \otimes V_{2j}\}_{i,j=1}^4$. The patterns of
$\{V_j\}_{j=1}^{16}$ can be found in
Appendix~\ref{App:Kraus_double_qubit}.

The Kraus error models as in \eqref{eq:Kraus_sq} and \eqref{eq:Kraus_dq}
are closed under composition operation, i.e., the composition of two Kraus
error models in the form of \eqref{eq:Kraus_sq} can be represented as a
new Kraus error model, also in the form of \eqref{eq:Kraus_sq}.
Lemma~\ref{lem:Kraus-composition-sq} and
Lemma~\ref{lem:Kraus-composition-dq} shows the composition properties of
single-qubit and double-qubit Kraus error model, respectively.

\begin{lemma} \label{lem:Kraus-composition-sq}
    Suppose two single-qubit Kraus errors $K_1$ and $K_2$ are in the form of
    \eqref{eq:Kraus_sq} with matrices and coefficients being denoted as
    $\{V_{ji}\}_{i=1}^4$ and $\{ a_{ji}, b_{ji} \}_{i=1}^3$ respectively
    for $j=1,2$. We further assume that 
    \begin{equation} \label{eq:assume-for-composition}
        a_{13}^2 + a_{23}^2 \le 1, \quad b_{13}^2 + b_{23}^2 \le 1.
    \end{equation}
    Then the composition of $K_1$ and $K_2$,
    $K = K_2 \circ K_1$, is another Kraus error in the form of
    \eqref{eq:Kraus_sq}. 
\end{lemma}

\begin{lemma}
    \label{lem:Kraus-composition-dq}
    Suppose two double-qubit Kraus errors $K_1 = K_{11} \otimes K_{21}$
    and $K_2 = K_{12}\otimes K_{22}$ are in the form of
    \eqref{eq:Kraus_dq}. We further assume that the parameters for both
    pairs $K_{11}, K_{12}$, and $K_{21}, K_{22}$ satisfy
    \eqref{eq:assume-for-composition}. Then the composition of $K_1$ and
    $K_2$, $K = K_2 \circ K_1$, is another Kraus error in the form of
    \eqref{eq:Kraus_dq}. 
\end{lemma}

Proofs of Lemma~\ref{lem:Kraus-composition-sq} and
Lemma~\ref{lem:Kraus-composition-dq} can be found in
Appendix~\ref{app:error-model-composition-lemmas}. In Kraus error models
as \eqref{eq:Kraus_sq} and \eqref{eq:Kraus_dq}, $V_j$ for $j \geq 2$ terms
are viewed as errors, and hence, the parameters $a_i$ and $b_i$ for $i
\geq 2$ are close to $0$. Thus the assumptions of $a_i$ and $b_i$ are
found reasonable in practice.

Other than the gate error models discussed above, the thermal relaxation
error model, which describes how errors may occur as time goes by, can
also be expressed in either the probabilistic error model or the Kraus
error model. Suppose a thermal relaxation channel is parametrized by
relaxation time constant $T_1, T_2$, gate time $t$, and excited state
thermal population $p_1$. If $T_1 < T_2$, then the thermal relaxation
channel can be expressed as a Kraus channel $\calE(\rho) = \sum_{j=1}^4
V_j \rho V_j^\dagger$, with
\begin{equation*}
    \begin{split}
        V_1 =
        \begin{pmatrix}
            1 - p_1 p_{\text{reset}} & 0 \\
            0 & e^{-t/T_2}
        \end{pmatrix}, & \quad
        V_2 =
        \begin{pmatrix}
            e^{-t/T_2} & 0 \\
            0 & 1 - p_0 p_{\text{reset}}
        \end{pmatrix}, \\
        V_3 =
        \begin{pmatrix}
            0 & 0 \\
            p_0 p_{\text{reset}} & 0
        \end{pmatrix}, & \quad
        V_4 =
        \begin{pmatrix}
            0 & p_1 p_{\text{reset}} \\
            0 & 0
        \end{pmatrix},
    \end{split}
\end{equation*}
where $p_0 = 1 - p_1$ and $p_{\text{reset}} = 1 - e^{-t/T_1}$. If $T_1 >
T_2$, then the channel is equivalent to a probabilistic model, with
$p_{R0} = p_0 p_{\text{reset}}, p_{R1} = p_1 p_{\text{reset}}, p_Z =
\frac{(1 - p_{\text{reset}})}{2} \bracket{1 - e^{\bracket{-t
(\frac{1}{T_2} - \frac{1}{T_1})}}}$, and $p_I = 1 - p_{R0} - p_{R1} -
p_Z$, where $R0$ and $R1$ represent the reset transformation to
$\ketbra{0}{0}$ and $\ketbra{1}{1}$, respectively. 

In Lemma~\ref{lem:prob-to-Kraus}, we show that some quantum errors, which
are introduced as probabilistic errors, could be represented as Kraus
errors as well.

\begin{lemma}
    \label{lem:prob-to-Kraus}
    Assume the probabilities of $X$ error and $Y$ error are equal, i.e.,
    $p_X = p_Y$, then any combination of $X$ error, $Y$ error, $Z$ error,
    reset to $\ketbra{0}{0}$, reset to $\ketbra{1}{1}$, and depolarizing
    error can be written as a Kraus error in the form of
    \eqref{eq:Kraus_sq}. 
\end{lemma}

Proof of Lemma~\ref{lem:prob-to-Kraus} can be found in
Appendix~\ref{app:error-model-composition-lemmas}. 

All mentioned errors can be described by either the Kraus error model
\eqref{eq:Kraus_sq} or the probabilistic error model. We summarize the
error models and various types of errors in Table~\ref{tab:errortype}. 

\begin{table}[htb]
    \centering
    \begin{tabular}{ccc}
        \toprule
        Error type & Probabilistic error model & Kraus error model
        \eqref{eq:Kraus_sq} \\
        \toprule
        $X$ & \Checkmark & \XSolidBrush \\
        $Y$ & \Checkmark & \XSolidBrush \\
        $Z$ & \Checkmark & \Checkmark  \\
        Reset to $\ketbra{0}{0}$ & \Checkmark & \Checkmark \\
        Reset to $\ketbra{0}{0}$ & \Checkmark & \Checkmark \\
        Depolarizing & \Checkmark & \Checkmark \\
        Amplitude damping & \XSolidBrush & \Checkmark \\
        Phase damping & \XSolidBrush & \Checkmark \\
        Thermal relaxation & \XSolidBrush & \Checkmark \\
        \midrule
        Combination & \XSolidBrush & Conditioned \\
        \bottomrule
    \end{tabular}
    \caption{The errors each error model can describe respectively. The
    combination of errors can be described by the Kraus error model in the
    form of \eqref{eq:Kraus_sq} under the condition that the probabilities
    of $X$ error and $Y$ error are equal, i.e. $p_X=p_Y$. }
    \label{tab:errortype}
\end{table}

\subsection{Linear growing error analysis}

We consider a quantum circuit with a sequence of $m$ quantum gates $G_1,
\cdots, G_m$ and their corresponding unitary matrices being $U_1, \cdots,
U_m$. We further denote $\rho_0$ as the initial density matrix, and
$\rho_k$ as the error-free density matrix after $k$ gates have been
implemented, i.e.,
\begin{equation} \label{eq:error_free}
    \rho_k = U_k \cdots U_1 \rho_0 U_1^\dagger \cdots U_k^\dagger.
\end{equation}
The actual density matrix with quantum error after $k$ gates is denoted as
$\tilde{\rho}_k$, i.e.,
\begin{equation*}
    \begin{split}
        \tilde{\rho}_0 = & \rho_0, \\
        \tilde{\rho}_1 = & E_1\bracket{U_1 \tilde{\rho}_0 U_1^\dagger}, \\
        \cdots & \\
        \tilde{\rho}_m = & E_m\bracket{U_m \tilde{\rho}_{m-1} U_m^\dagger}, 
    \end{split}
\end{equation*}
where $E_i$ is either a probabilistic error operator $P$, a Kraus error
operator $K$, or their composition. Notice that $\tilde{\rho}_k$ has
randomness if there are probabilistic errors.

A linearly growing bound for both Kraus error and probabilistic error can
be proved easily. We first give
Lemma~\ref{lem:linear-bound-for-Kraus-error} and
Lemma~\ref{lem:linear-bound-for-probabilistic-error} for Kraus error model
and probabilistic error model respectively. Then,
Lemma~\ref{lem:linear-bound-for-mixed-error} gives the linear growing
bound for mixed errors.

\begin{lemma}
    \label{lem:linear-bound-for-Kraus-error}
    Given a quantum circuit with a sequence of quantum gates $G_1, \cdots,
    G_m$ starting with an initial density matrix $\rho_0$. Each gate $G_k$
    is implemented with a Kraus error $K_k$. Suppose there is a constant
    $\gamma_1 > 0$ such that 
    \begin{equation} \label{eq:assume-for-Kraus-linear-bound}
        \fnorm{K_k(\rho) - \rho} \le \gamma_1,
    \end{equation}
    for any density matrix $\rho$ and $k = 1, \cdots, m$. Then the error
    of density matrix grows at most linearly,
    \begin{equation} \label{eq:linear_bound_for_Kraus}
        \fnorm{\tilde{\rho}_m - \rho_m} \le \gamma_1 m. 
    \end{equation}
\end{lemma}

Notice that Lemma~\ref{lem:linear-bound-for-Kraus-error} shows the error
growing for the general Kraus error model. A probabilistic error $P$ with
probability $p$ can be expressed as a general Kraus error,
\begin{equation}\label{eq:prob-to-Kraus}
    P(\rho) = p\rho + (1-p) R \rho R^\dagger
    = V_1 \rho V_1^\dagger + V_2 \rho V_2^\dagger, 
\end{equation}
where $V_1 = \sqrt{p}I$, $V_2 = \sqrt{1-p}R$, and unitary $R$ maps $\rho$
to another state. Therefore, we have the following results, which could be
viewed as corollaries of Lemma~\ref{lem:linear-bound-for-Kraus-error}.

\begin{lemma}
    \label{lem:linear-bound-for-probabilistic-error}
    Given a quantum circuit with a sequence of quantum gates $G_1, \cdots,
    G_m$ starting with an initial density matrix $\rho_0$. Each gate $G_k$
    is implemented with a probabilistic error $P_k$ with probability $0 <
    p_k \leq 1$. Then there is a constant $\gamma_2 > 0$ such that  
    \begin{equation}\label{eq:linear-bound-for-prob}
        \bbE \fnorm{\tilde{\rho}_m - \rho_m} \le \gamma_2 m. 
    \end{equation}
\end{lemma}

\begin{lemma}
    \label{lem:linear-bound-for-mixed-error}
    Given a quantum circuit with a sequence of quantum gates $G_1, \cdots,
    G_m$ starting with an initial density matrix $\rho_0$. Each gate $G_k$
    is implemented with either a probabilistic error with probability $0 <
    p_k \leq 1$, a Kraus error satisfying
    \eqref{eq:assume-for-Kraus-linear-bound}, or a mix of both. Then there
    is a constant $\gamma \ge 0$ independent of the number of qubits, such
    that
    \begin{equation*}
        \bbE \fnorm{\tilde{\rho}_m - \rho_m} \le \gamma m.
    \end{equation*}
\end{lemma}

The proofs of Lemma~\ref{lem:linear-bound-for-Kraus-error} and
Lemma~\ref{lem:linear-bound-for-probabilistic-error} are given in
Appendix~\ref{app:linear-error-bound}.
Lemma~\ref{lem:linear-bound-for-mixed-error} could be viewed as a simple
composition of Lemma~\ref{lem:linear-bound-for-Kraus-error} and
Lemma~\ref{lem:linear-bound-for-probabilistic-error}, and is stated
without detailed proof.

However, a linear growing upper bound does not agree well with
experimental results in Figure~\ref{fig:QFT-dem}. A simple inequality for
the Frobenius norm of the difference between two density matrices
indicates that the quantum error should be upper bounded by a constant,
i.e.,
\begin{equation} \label{eq:direct_bound}
    \fnorm{\tilde{\rho} - \rho}^2 = \fnorm{\tilde{\rho}}^2
    + \fnorm{\rho}^2 - 2\tr{\tilde{\rho}\rho} \le 2,
\end{equation}
for any density matrices $\tilde{\rho}$ and $\rho$, where the last
inequality is based on the fact that the product of two semi-definite
matrices has non-negative trace~\cite{Ruhe1970}. Therefore, the linear
growing bound cannot characterize the error propagates as the number of
gates increases. A tighter bound for quantum error propagations is
desired.

\section{Analysis of Kraus Error Propagation}
\label{sec:kraus-bound}

In this section, we focus on the analysis of quantum error propagation for
quantum circuits with Kraus error as in the form of \eqref{eq:Kraus_sq}
and \eqref{eq:Kraus_dq} only. The setting could be fairly similar to that
in Lemma~\ref{lem:linear-bound-for-Kraus-error}. But conclusions are
dramatically different. In this section, we show that the error
propagation would scale as $1-(1-q)^m$, where $m$ is the depth and $q$ is
a constant, $0 \leq q < 1$, depending on Kraus error parameters. The major
result of this section is given in Theorem~\ref{thm:thm_Kraus}.
Lemma~\ref{lem:Kraus-lemma-sq}, Lemma~\ref{lem:Kraus-lemma-dq}, and
Lemma~\ref{lem:Kraus} are proposed and proved to facilitate the proof of
Theorem~\ref{thm:thm_Kraus}.

\begin{theorem}
    \label{thm:thm_Kraus}
    Given a quantum circuit with a sequence of single-qubit or
    double-qubit gates $G_1, \cdots, G_m$ starting with an initial density
    matrix $\rho_0$. Each gate $G_k$ is implemented with a Kraus error
    $K_k$ as in form \eqref{eq:Kraus_sq} or \eqref{eq:Kraus_dq}. Then the
    error propagation is bounded as,
    \begin{equation} \label{eq:Kraus-thm}
        \fnorm{\tilde{\rho}_m - \rho_m}^2 \le 2(1 - (1 - q)^m),
    \end{equation}
    where $q$ is a constant, $0 \le q < 1$, independent of the number of
    qubits.
\end{theorem}

Before giving precise proof of Theorem~\ref{thm:thm_Kraus}, we first
sketch the key ideas therein. The basic idea is to find a constant $q \in
[0,1)$, for $q$ as small as possible, such that
\begin{equation} \label{eq:estimation}
    \fnorm{\tilde{\rho}_k - \rho_k}^2 \le (1-q)\fnorm{\tilde{\rho}_{k-1}
    - \rho_{k-1}}^2 + 2q,
\end{equation}
where $\tilde{\rho}_k$ and $\rho_k$ are noisy and noiseless density matrix
after acting $k$ gates, respectively. Then, we could recursively apply
\eqref{eq:estimation}, and obtain,
\begin{equation} \label{eq:goal}
    \begin{split}
        \fnorm{\tilde{\rho}_m - \rho_m}^2
        & \le (1-q)\fnorm{\tilde{\rho}_{m-1} - \rho_{m-1}}^2 + 2q \\
        & \le (1-q)\bracket{(1-q) \fnorm{\tilde{\rho}_{m-2}
        - \rho_{m-2}}^2 + 2q} + 2q \\
        & \le \cdots \\
        & \le 2q\bracket{1 + (1-q) + \cdots + (1-q)^{m-1}} \\
        & = 2(1-(1-q)^m),
    \end{split}
\end{equation}
which is the conclusion of Theorem~\ref{thm:thm_Kraus}. Substituting the
gate and quantum error action in the matrix form, \eqref{eq:estimation} is
equivalent to
\begin{equation} \label{eq:estimation_1}
    \begin{split}
        \fnorm{K_k\bracket{U_k \tilde{\rho}_{k-1} U_k^\dagger}
        - U_k \rho_{k-1} U_k^\dagger}^2
        & \le (1-q)\fnorm{\tilde{\rho}_{k-1} - \rho_{k-1}}^2 + 2q \\
        & = (1-q)\fnorm{U_k \tilde{\rho}_{k-1} U_k^\dagger
        - U_k \rho_{k-1} U_k^\dagger}^2 + 2q.
    \end{split}
\end{equation}
Denoting $\tilde{\rho} = U_k \tilde{\rho}_{k-1} U_k^\dagger$ and $\rho =
U_k \rho_{k-1} U_k^\dagger$, \eqref{eq:estimation_1} can, then, be written
as
\begin{equation} \label{eq:estimation_2}
    \fnorm{K_k\bracket{\tilde{\rho}} - \rho}^2
    \le (1-q)\fnorm{\tilde{\rho} - \rho}^2 + 2q. 
\end{equation}

We define a function $F$ of Kraus error $K$, density matrices
$\tilde{\rho}$ and $\rho$ as
\begin{equation} \label{eq:Fdef}
    F(K; \tilde{\rho}, \rho) =
    \begin{cases}
        \frac{\fnorm{K\bracket{\tilde{\rho}} - \rho}^2
        - \fnorm{\tilde{\rho} - \rho}^2}
        {2 - \fnorm{\tilde{\rho} - \rho}^2},
        & \fnorm{\tilde{\rho} - \rho}^2 < 2 \\
        \lim_{\tilde{\rho}^\prime \rightarrow \tilde{\rho},
        \rho^\prime \rightarrow \rho}
        F(K; \tilde{\rho}^\prime, \rho^\prime),
        & \fnorm{\tilde{\rho} - \rho}^2 = 2
    \end{cases}.
\end{equation}
Then, finding a constant $q$, as small as possible, satisfying
\eqref{eq:estimation_2} could be addressed by finding an upper bound of
$F(K; \tilde{\rho}, \rho)$, i.e.,
\begin{equation} \label{eq:supF}
    q = \sup_{\tilde{\rho}, \rho} F(K; \tilde{\rho}, \rho).
\end{equation}
Since $\fnorm{K\bracket{\tilde{\rho}} - \rho}^2 \le 2$ (see
\eqref{eq:ineq2}), it always holds $F(K; \tilde{\rho}, \rho) \le 1$.
However, this is not sufficient for us to estimate a tighter upper bound
for $F$. To achieve our goal, we need $F(K; \tilde{\rho}, \rho)$ to be
strictly less than $1$ for all $\tilde{\rho}$ and $\rho$, which is
guaranteed by the following lemmas. 

\begin{lemma}
    \label{lem:Kraus-lemma-sq}
    Suppose $K_\text{sq}$ is a single-qubit Kraus operator in the form of
    \eqref{eq:Kraus_sq}. Then there exist a constant $\delta > 0$
    independent of $n$, such that
    \begin{equation*}
        \fnorm{K_\text{sq}(\rho)}^2 \le 1 - \delta
    \end{equation*}
    for any $n$-qubit density matrix $\rho \in \mathbb{C}^{N \times N}$.
\end{lemma}

\begin{lemma}
    \label{lem:Kraus-lemma-dq}
    Suppose $K_\text{dq}$ is a double-qubit Kraus operator in the form of
    \eqref{eq:Kraus_dq}. Then there is a constant $\delta > 0$ independent
    of $n$, such that
    \begin{equation*}
        \fnorm{K_\text{dq}(\rho)}^2 \le 1 - \delta 
    \end{equation*}
    for any $n$-qubit density matrix $\rho \in \mathbb{C}^{N \times N}$.
\end{lemma}

In both Lemma~\ref{lem:Kraus-lemma-sq} and Lemma~\ref{lem:Kraus-lemma-dq},
we do not have explicit expressions for $\delta$. We only prove the
existence of $\delta$ satisfying $\delta > 0$. Both proofs of
Lemma~\ref{lem:Kraus-lemma-sq} and Lemma~\ref{lem:Kraus-lemma-dq} obey the
following flow. First, by triangle inequality, it can be proved that
$\fnorm{K(\rho)} \le 1$. Then we illustrate that the equality condition
cannot hold for $K_\text{sq}$ and $K_\text{dq}$, i.e.,
$\fnorm{K\bracket{\rho}} < 1$ for both $K_\text{sq}$ and $K_\text{dq}$.
Moreover, since $K$ is a continuous function of density matrix $\rho$,
which is defined on a compact set, the maximum of this function can be
achieved. Therefore, the upper bound of $\fnorm{K\bracket{\rho}}$ is
strictly less than $1$, for both $K_\text{sq}$ and $K_\text{dq}$.
Appendix~\ref{app:Kraus-lemma} proves Lemma~\ref{lem:Kraus-lemma-sq} and
Lemma~\ref{lem:Kraus-lemma-dq} in detail.

With Lemma~\ref{lem:Kraus-lemma-sq} and Lemma~\ref{lem:Kraus-lemma-dq}, we
then have the following lemma. 

\begin{lemma}
    \label{lem:Kraus}
    Suppose $K$ is a single-qubit Kraus operator in the form of
    \eqref{eq:Kraus_sq} or a double-qubit Kraus operator in the form of
    \eqref{eq:Kraus_dq}. There is a constant $0\le q < 1$, such that
    \begin{equation} \label{eq:lemma-Kraus}
        F(K; \tilde{\rho}, \rho) \le q,
    \end{equation}
    for any density matrices $\tilde{\rho}$ and $\rho$. 
\end{lemma}

\begin{proof}

According to Lemma~\ref{lem:Kraus-lemma-sq} or
Lemma~\ref{lem:Kraus-lemma-dq}, there is a constant $\delta>0$ such that
$\fnorm{K(\rho)}^2 \le 1 - \delta$ for any density matrix $\rho$.
Therefore, it holds
\begin{equation*}
    \fnorm{K(\tilde{\rho}) - \rho}^2 = \fnorm{K(\tilde{\rho})}^2
    + \fnorm{\rho}^2 - 2\tr{K(\tilde{\rho})\rho} \le 2 - \delta,
\end{equation*}
where the second inequality adopts the positivity of $K(\tilde{\rho})$ and
Ruhe's trace inequality. Thus, for any density matrices $\tilde{\rho}$ and
$\rho$ such that $\fnorm{\tilde{\rho} - \rho}^2 \ne 2$, it holds
\begin{equation*}
    F(K; \tilde{\rho}, \rho) = \frac{\fnorm{K(\tilde{\rho}) - \rho}^2
    - \fnorm{\tilde{\rho} - \rho}^2}{2 - \fnorm{\tilde{\rho} - \rho}^2}
    \le \frac{2 - \delta}{2}. 
\end{equation*}
When $\fnorm{\tilde{\rho} - \rho}^2 = 2$, it holds
\begin{equation*}
    F(K; \tilde{\rho}, \rho) =
    \lim_{\tilde{\rho}^\prime \rightarrow \tilde{\rho},
    \rho^\prime \rightarrow \rho}
    F(K; \tilde{\rho}^\prime, \rho^\prime) \le \frac{2 - \delta}{2}. 
\end{equation*}
Setting $q = \frac{2 - \delta}{2} < 1$, we have \eqref{eq:lemma-Kraus}. 

\end{proof}

Finally, we prove Theorem~\ref{thm:thm_Kraus} using Lemma~\ref{lem:Kraus}.

\begin{proof}[Proof of Theorem~\ref{thm:thm_Kraus}]

For each Kraus operator $K_k$, according to Lemma~\ref{lem:Kraus}, there
is a constant $0 \le q_k < 1$ such that $F(K_k; \tilde{\rho}, \rho) \le
q_k$, for any density matrices $\tilde{\rho}$ and $\rho$. Let $q =
\max\left\{q_1, \cdots, q_m\right\} < 1$. Following the derivations from
\eqref{eq:estimation} to \eqref{eq:supF}, \eqref{eq:Kraus-thm} is proved. 

\end{proof}

\section{Analysis of Mixed Error Propagation}
\label{sec:prob-mix-bound}

Following the flow in Section~\ref{sec:preliminary}, we first prove the
expected error propagation for a quantum circuit with only probabilistic
errors, as in Theorem~\ref{thm:thm_Prob}. Then, we combine the results of
Theorem~\ref{thm:thm_Kraus} and Theorem~\ref{thm:thm_Prob} to show that
the error propagation of a quantum circuit with both Kraus errors and
probabilistic errors can be bounded by $2 (1 - (1 - r)^m)$, where $m$ is
the number of gates and $r$ is a constant in $[0,1)$ independent of $m$.
The combined result is known as the error propagation of mixed error
models, which is formally stated in Theorem~\ref{thm:main_thm}.

\subsection{Analysis of Probabilistic Error Propagation}

For a quantum circuit with only probabilistic errors, we could write the
expected error propagation into two parts: at least one error occurs and
no error occurs. Since the difference between any two density matrices is
bounded by a constant, the ``at least one error occurs'' part is bounded
by its probability, which scales as $(1 - (1-p)^m)$. The ``no error
occurs'' part does not contribute to the error propagation and is omitted
directly. Put two together, we obtain the Theorem~\ref{thm:thm_Prob} for
the quantum circuit with probabilistic error only.

\begin{theorem}
    \label{thm:thm_Prob}
    Given a quantum circuit with a sequence of single-qubit or
    double-qubit gates $G_1, \cdots, G_m$ starting with an initial density
    matrix $\rho_0$. Each gate $G_k$ is implemented with a probabilistic
    error $P_k$ with error probability $p_k \in [0, 1)$. Then the expected
    error propagation is bounded as,
    \begin{equation*}
        \bbE \fnorm{\tilde{\rho}_m - \rho_m}^2 \le 2(1 - (1 - p)^m),
    \end{equation*}
    where $p$ is a constant, $0 \le p < 1$, independent of the number of
    qubits.
\end{theorem}

\begin{proof}

If at least one error occurs in applying $G_1, \cdots, G_m$, then we use
\eqref{eq:direct_bound} to bound $\fnorm{\tilde{\rho}_m - \rho_m}^2 \leq 2$.
Denote $p=\max_{1 \le k \le m}p_k$.
Then expected error propagation can be bounded as,
\begin{equation*}
    \begin{split}
        \bbE \fnorm{\tilde{\rho}_m - \rho_m}^2
        & \leq 2 \cdot \bbP(\text{at least one error occurs})
        + \fnorm{\rho_m - \rho_m}^2 \cdot \bbP(\text{no error}) \\
        & = 2 \bracket{1-\prod_{k=1}^m \bracket{1-p_k}}
        \le 2\bracket{1-\bracket{1-p}^m},
    \end{split}
\end{equation*}
where the probability of no error is $\prod_{k=1}^m \bracket{1-p_k}$, and
that at least one error occurs is $1-\prod_{k=1}^m \bracket{1-p_k}$.

\end{proof}

\subsection{Proof of Mixed Error Propagation}

Combining Theorem~\ref{thm:thm_Kraus} and Theorem~\ref{thm:thm_Prob}
together, we can then prove Theorem~\ref{thm:main_thm}, which
characterizes the expected error propagation of mixed error. The proof
flow can also be adapted to prove
Lemma~\ref{lem:linear-bound-for-mixed-error}.

\begin{proof}[Proof of Theorem~\ref{thm:main_thm}]

Suppose there are $m_K$ gates with Kraus errors and $m_P$ gates with
probabilistic errors, with error probability being $p_1,\cdots,p_{m_P}$.
If no probabilistic error occurs, of which the probability is
$\bbP(\text{no error}) = \prod_{k=1}^{m_P}\bracket{1-p_k}$, then the
system is reduced to a circuit with only Kraus errors. By Theorem
\ref{thm:thm_Kraus}, there is a constant $q \in [0,1)$ such that,
\begin{equation*}
    \fnorm{\tilde{\rho}_m - \rho_m}^2 \le 2\bracket{1-\bracket{1-q}^{m_K}}. 
\end{equation*}

We then adopt the decomposition as in the proof of
Theorem~\ref{thm:thm_Prob}. The expected error propagation is bounded as,
\begin{equation*}
    \begin{split}
        \bbE \fnorm{\tilde{\rho}_m - \rho_m}^2
        & \le 2 \cdot \bbP(\text{at least one error occurs})
        + \fnorm{\tilde{\rho}_m - \rho_m}^2 \cdot \bbP(\text{no error}) \\
        & \le 2 \bracket{1-\prod_{k=1}^{m_P}\bracket{1-p_k}} 
        + 2\bracket{1-\bracket{1-q}^{m_K}}\prod_{k=1}^{m_P}\bracket{1-p_k} \\
        & = 2\bracket{1 - \bracket{1-q}^{m_K}\prod_{k=1}^{m_P}\bracket{1-p_k}}
        \le 2\bracket{1-\bracket{1-r}^m}, 
    \end{split}
\end{equation*}
where $r = \max\left\{p_1,\cdots,p_{m_P},q\right\}<1$. 

\end{proof}

\section{Numerical Experiments}
\label{sec:numerical}

In this section, we simulate quantum circuits using the quantum device
backends provided by IBM Quantum Experience. In these quantum device
backends, both the Kraus error model and the probabilistic error model are
adopted to represent different types of quantum errors. Hence, these
backend simulators are well characterized by our mixed error models as in
Theorem~\ref{thm:main_thm}. We also simulate quantum circuits to validate
Theorem~\ref{thm:thm_Kraus} for the Kraus errors and
Theorem~\ref{thm:thm_Prob} for probabilistic errors separately with each
of these error models.

\emph{Kraus Error.}
We first simulate a single-qubit circuit consisting of identity gates with
only Kraus error in the form of \eqref{eq:Kraus_sq} with parameters being
$a_2 = b_2 = 0.001$, $a_3 = 0.08$, $b_3 = 0.008$, $a_1 = \sqrt{1 - a_2^2 -
a_3^2} \approx 0.9968$, and $b_1 = \sqrt{1 - b_2^2 - b_3^2} = 0.99997$. We
simulate the setup with various circuit depths from 100 to 2000. The
Frobenius norm of the density matrix difference is measured every $100$
gate, demonstrating the error propagation. We plot both the empirical
error propagation as well as the upper bound given in
Theorem~\ref{thm:thm_Prob} in Figure~\ref{fig:Kraus}. The constant $q$ in
Theorem~\ref{thm:thm_Prob} is not explicitly given. For the above Kraus
error model setup, we estimate $q$ by random sampling $10^8$ pairs of
density matrices $\tilde{\rho}$ and $\rho$ and find the maximum value of
$F(K; \tilde{\rho}, \rho)$. The matrices are generated by $R\Lambda R^T$,
where $R$ is a rotation matrix with uniformly sampled rotation angle, and
$\Lambda$ is uniformly sampled positive semi-definite diagonal matrix with
trace being $1$. The estimated constant $q$ is $5.620 \times 10^{-3}$. The
result is shown in Figure~\ref{fig:Kraus}. 

\begin{figure}[htb]
    \centering
    \includegraphics[width=0.6\textwidth]{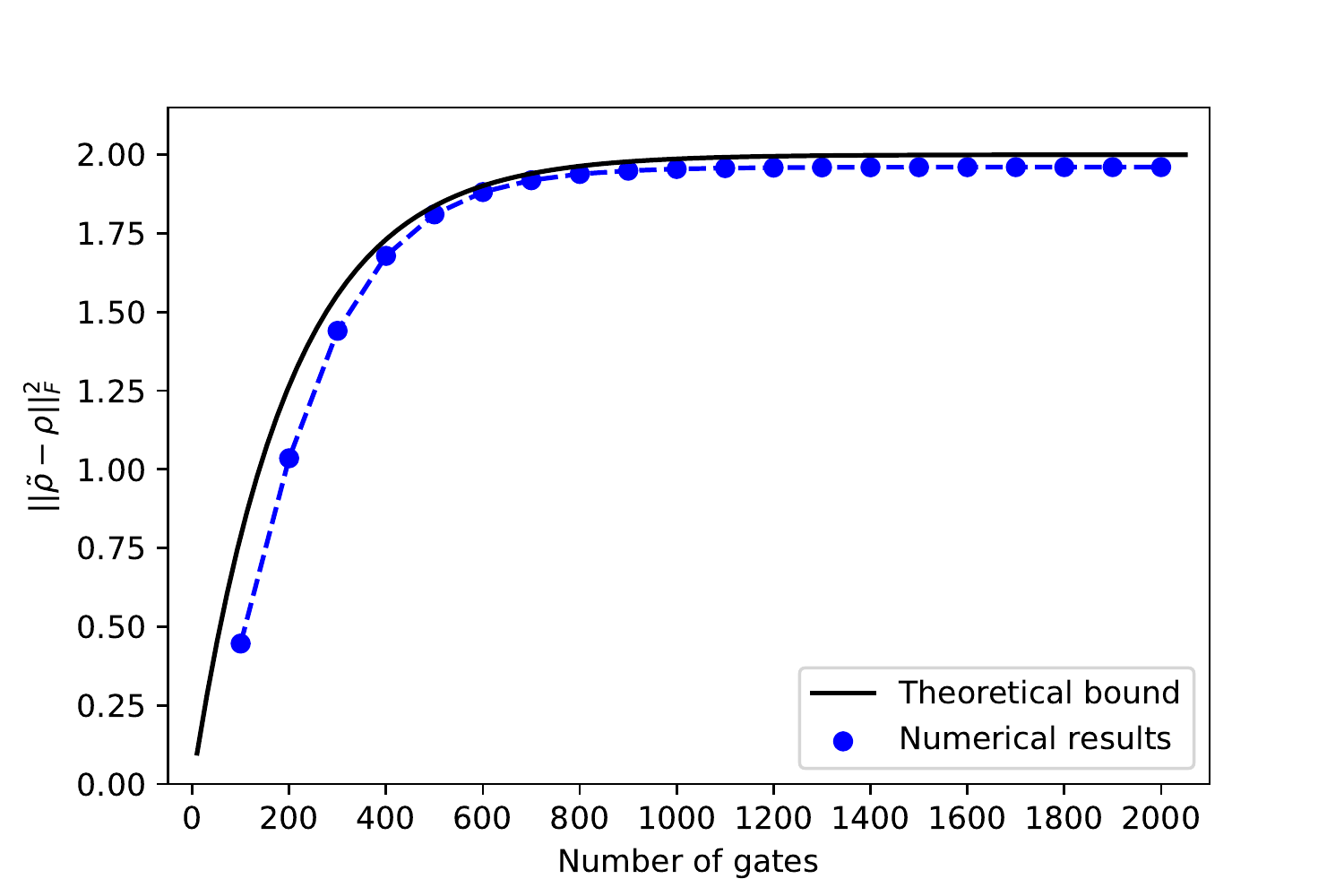}\\
    \caption{Numerical experiments for identity circuit with Kraus errors
    only. A quantum circuit consisting of single-qubit identity gates is
    simulated using 'qasm' backend. Each gate is implemented with a Kraus
    error. The resulting density matrix is measured for every $100$ gate.
    }\label{fig:Kraus}
\end{figure}

Theorem~\ref{thm:thm_Kraus} provides an upper bound for the error
propagation. As shown in Figure~\ref{fig:Kraus}, with the estimated $q$,
the numerical errors are bounded by our theoretical bound, and the two
curves are fairly close to each other. Since our result is a worst-case
upper bound, from Figure~\ref{fig:Kraus}, we are confident that our
analysis provides a fairly tight upper bound.

\emph{Probabilistic Error.}
We then simulate a single-qubit circuit consisting of identity gates with
only probabilistic error using `qasm' backend. For each gate, there is a
probability $p = 0.005$ for reset-to-$\ketbra{1}{1}$ error. Similarly, we
simulate the setup with various circuit depths from 100 to 2000. The
Frobenius norm of the density matrix difference is measured every $100$
gate, demonstrating the error propagation. Each circuit is repeatedly
executed $8192$ times to approximate the expectation value of the error.
We plot both the empirical expected error propagation as well as the upper
bound given in Theorem~\ref{thm:thm_Prob} in Figure~\ref{fig:Prob}.

\begin{figure}[htb]
    \centering
    \includegraphics[width=0.6\textwidth]{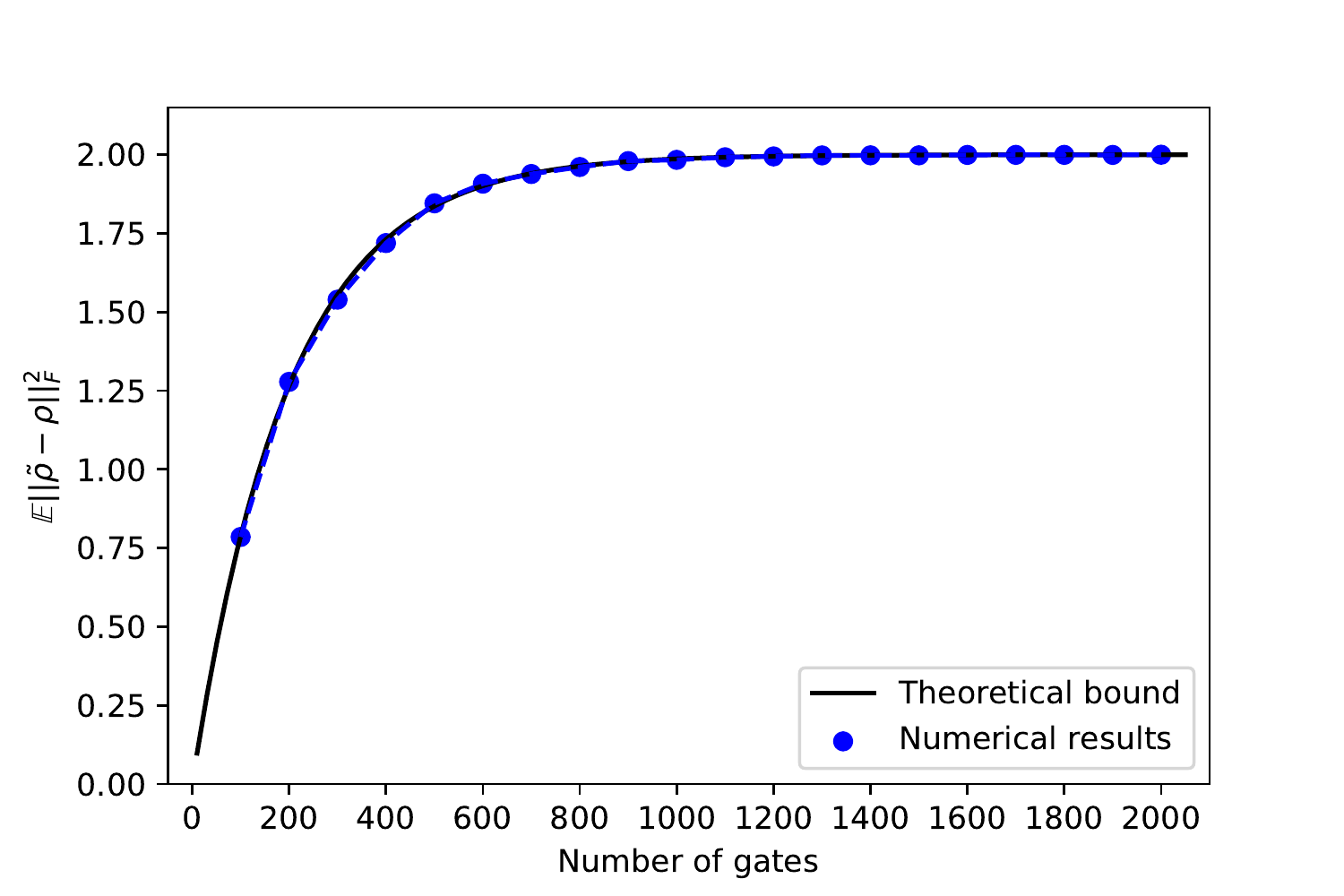}\\
    \caption{Numerical experiments for single-qubit identity circuit with
    probabilistic errors only. Each gate is implemented with a
    reset-to-$\ketbra{1}{1}$ error with probability $p=0.005$. Each
    circuit is repeatedly executed for $8192$ to approximate the error
    expectation value.} \label{fig:Prob}
\end{figure}

Our analysis results in Theorem~\ref{thm:thm_Prob} provides an upper bound
for the error expectation value, which is not necessarily an upper bound
for the sample mean of the error. While, as shown in
Figure~\ref{fig:Prob}, the numerical errors are tightly bounded by our
theoretical bound. 

\emph{Mixed Error.}
Finally, we simulate a multi-qubit circuit with both Kraus error and
probabilistic errors. The base circuit is chosen to be the same as that in
Figure~\ref{fig:QFT-dem}, i.e., $3$-, $4$-, and $5$-qubit QFT circuits.
The quantum simulator used in this case is the FakeVigo backend, which is
configured to simulate IBM Vigo quantum computer. For the purpose of
verifying our theoretical bound, we turn off the readout error and
measurement error, and only keep the gate errors. The constant in
Theorem~\ref{thm:thm_Prob} is read from the configuration, $p=1.076 \times
10^{-2}$. The constants for various Kraus errors in
Theorem~\ref{thm:thm_Kraus} are estimated separately in the same way as
that in the Kraus error numerical part, and the overall constant is found to be
$q \approx 4.769\times10^{-3}$. Thus, the constant $r$ in
Theorem~\ref{thm:main_thm} is set to be $r = \max\{p,q\} = 1.076 \times
10^{-2}$. Since the resulting upper bound in Theorem~\ref{thm:main_thm} is
independent of the number of qubits, we plot the theoretical upper bound for
all three multi-qubit circuits as the black curve in
Figure~\ref{fig:QFT-mixed}. 

\begin{figure}[htb]
    \centering
    \includegraphics[width=0.6\textwidth]{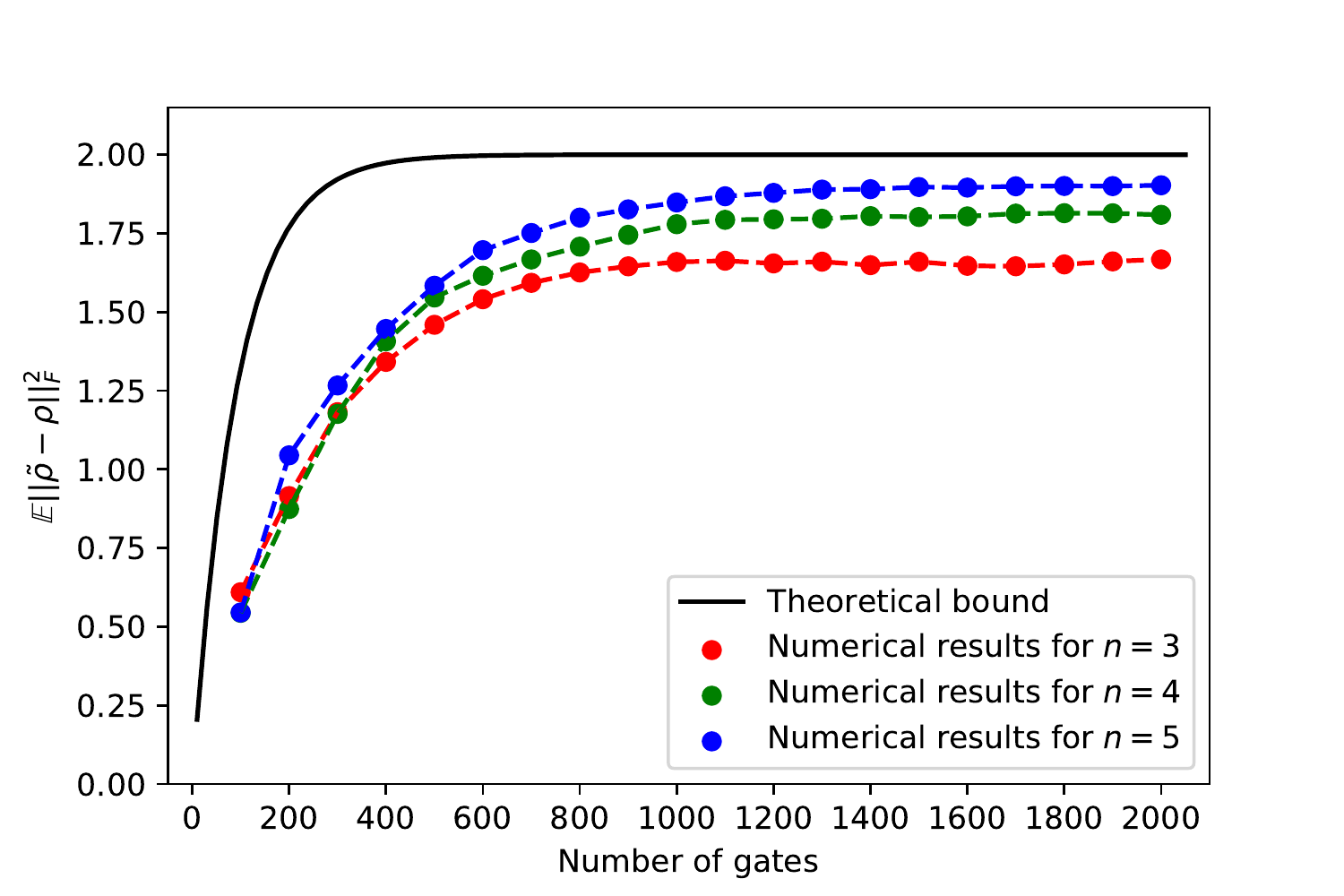}\\
    \caption{Numerical experiments for repeated $3$, $4$, and $5$-qubit
    QFT circuits on FakeVigo backend with only gate errors. Each circuit
    is repeatedly executed for $8192$ to approximate the error expectation
    value.} \label{fig:QFT-mixed}
\end{figure}

According to Figure~\ref{fig:QFT-mixed}, we find that our analysis result
in Theorem~\ref{thm:main_thm} is indeed an upper bound for all three
quantum circuits on 3-, 4-, and 5-qubit systems. Vigo is a 5-qubit quantum
computer. Different qubits in Vigo actually have different error profiles.
If we bind the quantum circuit to specific qubits and estimate the
constant separately, we could have three different theoretical upper
bounds. In this experiment, we simply estimate the constant for all 5
qubits together and obtain an upper bound for all quantum circuits on
Vigo. As we observe from Figure~\ref{fig:QFT-mixed}, quantum circuits with
a larger number of qubits lead to larger errors. Further, in the current
quantum computers, even for quantum computers with a small number of
qubits, the error quickly grows to a non-negligible level. Quantum
circuits with tens to hundreds of depths would be the limit on Vigo. For
quantum computers with a larger number of quantum qubits, numerically
simulating the density matrix error results is not possible due to the
exponentially increasing size of the density matrix. However, our
theoretical upper bound could still be calculated and provide a fairly
good estimation on the error growth.

\section{Conclusion}
\label{sec:conclusion}

Quantum computing is becoming a promising tool for computational tasks.
While the quantum hardware is not perfect and is expected to bear with
large noise for a long time. The performance of quantum computers
nowadays is limited by quantum gate errors and sampling errors. Therefore,
we aim to characterize the quantum error that grows with the number of gates in
a quantum circuit.

In this work, we first use traditional numerical analysis methods to prove
that quantum error grows linearly with the depth of a circuit. However, a
simple calculation of the density matrices suggests that the error could
not grow linearly but hit a plateau towards the end instead. We,
therefore, provided a more carefully analyzed upper bound that better
characterizes the growth of the quantum error with the depth of a circuit.
To be more specific, we analyzed the probabilistic error model, the Kraus
error model, and the mix of both. For all three cases, the error grows as
$\sim (1 - (1 - c)^m)$, where $c$ is a constant independent of qubit
number and circuit depth ($c$ is different for different error models),
and $m$ is the number of quantum gates in the circuit. Finally, we did
numerical experiments on the simulator of the Vigo quantum computer
provided by IBM Quantum Experience. Numerical results for identity quantum
circuits and QFT circuits suggested that our theoretical bound is tight.
The errors of QFT circuits are well controlled by our bound. 

An immediate future direction is to explore different metrics of the
quantum error propagations. We know that all metrics in a
finite-dimensional Hilbert space are equivalent. Hence, our results could
be directly extended to other metrics with an extra dimension dependent
constant. A more careful analysis could reduce such a constant. Another
future direction is to obtain a qualititive estimation of the constant $q$,
$p$, and $r$ in Theorem~\ref{thm:thm_Kraus}, Theorem~\ref{thm:thm_Prob},
and Theorem~\ref{thm:main_thm} respectively.

\bibliography{library}

\appendix

\gdef\thesection{\Alph{section}}
\makeatletter
\renewcommand\@seccntformat[1]{\appendixname\ \csname
the#1\endcsname.\hspace{0.5em}}
\makeatother

\section{Probabilistic error model in Kraus form}
\label{app:kraus-prob}

Kraus formulations of typical probabilistic errors are listed below. 

\begin{itemize}
    \item Bit flip with probability $p_X$: 
    \begin{equation*}
        V_1 = \sqrt{p_X} I = \sqrt{p_X}
        \begin{pmatrix}
            1 & 0 \\
            0 & 1
        \end{pmatrix}, 
        V_2 = \sqrt{1-p_X} X = \sqrt{1-p_X}
        \begin{pmatrix}
            0 & 1 \\
            1 & 0
        \end{pmatrix}. 
    \end{equation*}
    \item Phase flip with probability $p_Z$: 
    \begin{equation*}
        V_1 = \sqrt{p_Z} I = \sqrt{p_Z}
        \begin{pmatrix}
            1 & 0 \\
            0 & 1
        \end{pmatrix}, 
        V_2 = \sqrt{1-p_Z} Z = \sqrt{1-p_Z}
        \begin{pmatrix}
            1 & 0 \\
            0 & -1
        \end{pmatrix}. 
    \end{equation*}
    \item Bit-phase flip with probability $p_Y$: 
    \begin{equation*}
        V_1 = \sqrt{p_Y} I = \sqrt{p_Y}
        \begin{pmatrix}
            1 & 0 \\
            0 & 1
        \end{pmatrix}, 
        V_2 = \sqrt{1-p_Y} Y = \sqrt{1-p_Y}
        \begin{pmatrix}
            0 & -\imath \\
            \imath & 0
        \end{pmatrix}. 
    \end{equation*}
    \item Reset to $\ketbra{0}{0}$ with probability $p_{R0}$: 
    \begin{equation*}
        V_1 = \sqrt{p_{R0}}
        \begin{pmatrix}
            1 & 0 \\
            0 & 1
        \end{pmatrix}, 
        V_2 = \sqrt{1-p_{R0}}
        \begin{pmatrix}
            1 & 0 \\
            0 & 0
        \end{pmatrix}, 
        V_3 = \sqrt{1-p_{R0}}
        \begin{pmatrix}
            0 & 1 \\
            0 & 0
        \end{pmatrix}. 
    \end{equation*}
    \item Reset to $\ketbra{1}{1}$ with probability $p_{R1}$: 
    \begin{equation*}
        V_1 = \sqrt{p_{R1}}
        \begin{pmatrix}
            1 & 0 \\
            0 & 1
        \end{pmatrix}, 
        V_2 = \sqrt{1-p_{R1}}
        \begin{pmatrix}
            1 & 0 \\
            0 & 0
        \end{pmatrix}, 
        V_3 = \sqrt{1-p_{R1}}
        \begin{pmatrix}
            0 & 0 \\
            1 & 0
        \end{pmatrix}. 
    \end{equation*}
    \item Depolarizing with probability $p_D$: 
    \begin{equation*}
        V_1 = \sqrt{\frac{1-3p_D}{4}}I, 
        V_2 = \frac{\sqrt{p_D}}{2}X, 
        V_3 = \frac{\sqrt{p_D}}{2}Y, 
        V_4 = \frac{\sqrt{p_D}}{2}Z. 
    \end{equation*}
\end{itemize}

\section{Double-qubit Kraus error model}
\label{App:Kraus_double_qubit}

The detailed expressions of $V_j$s in \eqref{eq:Kraus_dq} are of forms,
\begin{equation*}\label{eq:Kraus_dq_matrix_form}
    \begin{split}
        & V_{1} =
        \begin{pmatrix}
            a_{1} & 0 & 0 & 0 \\
            0 & b_{1} & 0 & 0 \\
            0 & 0 & c_{1} & 0 \\
            0 & 0 & 0 & d_{1}
        \end{pmatrix}, \quad
        V_{2} =
        \begin{pmatrix}
            a_{2} & 0 & 0 & 0 \\
            0 & b_{2} & 0 & 0 \\
            0 & 0 & c_{2} & 0 \\
            0 & 0 & 0 & d_{2}
        \end{pmatrix}, \quad
        V_{3} =
        \begin{pmatrix}
            a_{3} & 0 & 0 & 0 \\
            0 & b_{3} & 0 & 0 \\
            0 & 0 & c_{3} & 0 \\
            0 & 0 & 0 & d_{3}
        \end{pmatrix}, \\
        & V_{4} =
        \begin{pmatrix}
            a_{4} & 0 & 0 & 0 \\
            0 & b_{4} & 0 & 0 \\
            0 & 0 & c_{4} & 0 \\
            0 & 0 & 0 & d_{4}
        \end{pmatrix}, \quad
        V_{5} =
        \begin{pmatrix}
            0 & 0 & 0 & 0 \\
            0 & 0 & 0 & 0 \\
            a_{5} & 0 & 0 & 0 \\
            0 & b_{5} & 0 & 0
        \end{pmatrix}, \quad
        V_{6} =
        \begin{pmatrix}
            0 & 0 & 0 & 0 \\
            0 & 0 & 0 & 0 \\
            a_{6} & 0 & 0 & 0 \\
            0 & b_{6} & 0 & 0
        \end{pmatrix}, \\
        & V_{7} =
        \begin{pmatrix}
            0 & 0 & 0 & 0 \\
            a_{7} & 0 & 0 & 0 \\
            0 & 0 & 0 & 0 \\
            0 & 0 & c_{7} & 0
        \end{pmatrix}, \quad
        V_{8} =
        \begin{pmatrix}
            0 & 0 & 0 & 0 \\
            a_{8} & 0 & 0 & 0 \\
            0 & 0 & 0 & 0 \\
            0 & 0 & c_{8} & 0
        \end{pmatrix}, \quad
        V_{9} =
        \begin{pmatrix}
            0 & b_{7} & 0 & 0 \\
            0 & 0 & 0 & 0 \\
            0 & 0 & 0 & d_{7} \\
            0 & 0 & 0 & 0
        \end{pmatrix}, \\
        & V_{10} =
        \begin{pmatrix}
            0 & b_{8} & 0 & 0 \\
            0 & 0 & 0 & 0 \\
            0 & 0 & 0 & d_{8} \\
            0 & 0 & 0 & 0
        \end{pmatrix}, \quad
        V_{11} =
        \begin{pmatrix}
            0 & 0 & c_{5} & 0 \\
            0 & 0 & 0 & d_{5} \\
            0 & 0 & 0 & 0 \\
            0 & 0 & 0 & 0
        \end{pmatrix}, \quad
        V_{12} =
        \begin{pmatrix}
            0 & 0 & c_{6} & 0 \\
            0 & 0 & 0 & d_{6} \\
            0 & 0 & 0 & 0 \\
            0 & 0 & 0 & 0
        \end{pmatrix}, \\
        & V_{13} =
        \begin{pmatrix}
            0 & 0 & 0 & 0 \\
            0 & 0 & 0 & 0 \\
            0 & 0 & 0 & 0 \\
            a_{9} & 0 & 0 & 0
        \end{pmatrix}, \quad
        V_{14} =
        \begin{pmatrix}
            0 & 0 & 0 & 0 \\
            0 & 0 & 0 & 0 \\
            0 & b_{9} & 0 & 0 \\
            0 & 0 & 0 & 0
        \end{pmatrix}, \quad
        V_{15} =
        \begin{pmatrix}
            0 & 0 & 0 & 0 \\
            0 & 0 & c_{9} & 0 \\
            0 & 0 & 0 & 0 \\
            0 & 0 & 0 & 0
        \end{pmatrix}, \\
        & V_{16} =
        \begin{pmatrix}
            0 & 0 & 0 & d_{9} \\
            0 & 0 & 0 & 0 \\
            0 & 0 & 0 & 0 \\
            0 & 0 & 0 & 0
        \end{pmatrix}
    \end{split}
\end{equation*}
The matrices $V_j$ obey normalization condition \eqref{eq:normalize},
which is equivalent to
\begin{equation*}
    \sum_{j=1}^{9} a_j^2 = 1, \quad
    \sum_{j=1}^{9} b_j^2 = 1, \quad
    \sum_{j=1}^{9} c_j^2 = 1, \quad
    \sum_{j=1}^{9} d_j^2 = 1.
\end{equation*}
Similar to the single-qubit Kraus error model, $V_1$ is close to an
identity matrix, while others are close to zero matrices, i.e., $a_1, b_1,
c_1, d_1$ are close to $1$, and other parameters are close to $0$. Also,
the parameters $a_j, b_j, c_j, d_j$ in \eqref{eq:Kraus_dq_matrix_form}
admits,
\begin{equation*}
    \det{
    \begin{pmatrix}
        a_1 & a_2 & a_3 & a_4 \\
        b_1 & b_2 & b_3 & b_4 \\
        c_1 & c_2 & c_3 & c_4 \\
        d_1 & d_2 & d_3 & d_4 \\
    \end{pmatrix}
    } \neq 0.
\end{equation*}

\section{Proof of error model composition lemmas}
\label{app:error-model-composition-lemmas}

We first prove Lemma~\ref{lem:system-of-eq-Kraus} to pave the path to the
proofs of other lemmas.

\begin{lemma}
    \label{lem:system-of-eq-Kraus}
    There exist a solution $a_1, a_2, b_1, b_2$ for a system of equations,
    \begin{equation} \label{eq:system-of-eq-Kraus}
        \begin{cases}
            a_1^2 + a_2^2 = s\\
            b_1^2 + b_2^2 = t\\
            a_1 b_1 + a_2 b_2 = r
        \end{cases}
    \end{equation}
    where $0 \leq s, t \leq 1$, and $s t \geq r^2$.
\end{lemma}

\begin{proof}

Due to the fact that $0 \leq s, t \leq 1$, we could parametrize $a_1, a_2,
b_1, b_2$ as
\begin{equation*}
    \begin{split}
        a_1 = \sqrt{s}\cos\theta, \quad
        a_2 = \sqrt{s}\sin\theta, \\
        b_1 = \sqrt{t}\cos\phi, \quad
        b_2 = \sqrt{t}\sin\phi. 
    \end{split}
\end{equation*}
The last equation in \eqref{eq:system-of-eq-Kraus}
then admits,
\begin{equation*}
    \sqrt{st} \cos(\theta - \phi) = r. 
\end{equation*}
Since $st \geq r^2$, we have $-1 \leq \frac{r}{\sqrt{st}} \leq 1$. Hence,
there is a unique solution for $\theta - \phi$. Therefore, the system of
equations \eqref{eq:system-of-eq-Kraus} has infinitely many solutions.

\end{proof}

\begin{proof}[Proof of Lemma~\ref{lem:Kraus-composition-sq}]
We denote
\begin{equation*}
    \begin{split}
        A_j &= 1-a_{j3}^2-b_{j3}^2, \\
        B_j &= b_{j3}^2, \\
        C_j &= a_{j1}b_{j1}+a_{j2}b_{j2},
    \end{split}
\end{equation*}
for $j = 1, 2$. The Kraus error model \eqref{eq:Kraus_sq} then could be
written as
\begin{equation*}
    K_j(\rho) =
    \begin{pmatrix}
        A_j r_{11} + B_j & C_j r_{12} \\
        C_j r_{21} & 1 - A_j r_{11} - B_j
    \end{pmatrix}, 
\end{equation*}
where
\begin{equation*}
    \rho =
    \begin{pmatrix}
        r_{11} & r_{12} \\
        r_{21} & r_{22}
    \end{pmatrix}, 
\end{equation*}
and $\tr{\rho} = 1$.
The composition of $K_1$ and $K_2$ admits,
\begin{equation*}
    K_2(K_1(\rho)) =
    \begin{pmatrix}
        A_1 A_2 r_{11} + B_1 A_2 + B_2 & C_1 C_2 r_{12} \\
        C_1 C_2 r_{21} & 1 - A_1 A_2 r_{11} - B_1 A_2 - B_2
    \end{pmatrix}. 
\end{equation*}

Therefore, by requiring the parameters $a_1,b_1,a_2,b_2,a_3,b_3$ satisfying
\begin{equation} \label{eq:compose}
    \begin{cases}
        A_1 A_2 = 1 - a_3^2 - b_3^2 = a_1^2 + a_2^2 + b_1^2 + b_2^2-1 \\
        B_1 A_2 + B_2 = b_3^2 = 1 - b_1^2 - b_2^2 \\
        C_1 C_2 = a_1 b_1 + a_2 b_2
    \end{cases},
\end{equation}
we are sure that the Kraus error is in the form of
\eqref{eq:Kraus_sq_matrix_form} and satisfies $K = K_2 \circ K_1$. We
rearrange \eqref{eq:compose} and obtain,
\begin{equation*}
    \begin{cases}
        a_1^2 + a_2^2 = A_1 A_2 + B_1 A_2 + B_2 \\
        b_1^2 + b_2^2 = 1 - B_1 A_2 - B_2 \\
        a_1 b_1 + a_2 b_2 = C_1 C_2
    \end{cases}.
\end{equation*}
We first notice that 
\begin{equation*}
    A_1 A_2 + B_1 A_2 + B_2 = 1 - (1 - a_{13}^2) a_{23}^2
    - (1 - b_{23}^2)a_{13}^2 \le 1,
\end{equation*}
and
\begin{equation*}
    A_1 A_2 + B_1 A_2 + B_2 \ge a_{13}^2 a_{23}^2 + a_{13}^2 b_{23}^2
    \ge 0,
\end{equation*}
where the second inequality adopts the assumption $a_{13}^2 + a_{23}^2 \le
1$.

We also have,
\begin{equation*}
    1 - B_1 A_2 - B_2 = 1 - b_{13}^2 (1 - a_{23}^2)
    - b_{23}^2 (1 - b_{13}^2) \le 1,
\end{equation*}
and
\begin{equation*}
    1 - B_1 A_2 - B_2 \ge b_{13}^2 a_{23}^2 + b_{13}^2 b_{23}^2 \ge 0,
\end{equation*}
where the second inequality adopts the assumption $b_{13}^2 + b_{23}^2 \le
1$.

Using the Cauchy inequality, we obtain
\begin{equation*}
    \begin{split}
        \bracket{C_1C_2}^2 \le & \bracket{a_{11}^2 + a_{12}^2}
        \bracket{b_{11}^2 + b_{12}^2} \bracket{a_{21}^2 + a_{22}^2}
        \bracket{b_{21}^2 + b_{22}^2} \\
        = & \bracket{1 - a_{13}^2} \bracket{1 - b_{13}^2}
        \bracket{1 - a_{23}^2} \bracket{1-b_{23}^2}, 
    \end{split}
\end{equation*}
and, hence,
\begin{equation*}
    \begin{split}
        & \bracket{A_1 A_2 + B_1 A_2 + B_2} \bracket{1 - B_1 A_2 - B_2}
        - \bracket{C_1C_2}^2 \\
        \ge & a_{13}^2 b_{23}^2 \bracket{1 - b_{13}^2 - b_{23}^2}
        + a_{23}^2 b_{13}^2 \bracket{1 - a_{13}^2 - a_{23}^2}
        + a_{13}^2 b_{13}^2 \bracket{a_{23}^4 + a_{23}^2 b_{23}^2
        + b_{23}^4} \ge 0. 
    \end{split}
\end{equation*}
Finally, by Lemma~\ref{lem:system-of-eq-Kraus}, we know that
\eqref{eq:compose} has solutions and $K_2 \circ K_1$ can be in the form of
\eqref{eq:Kraus_sq_matrix_form}.

\end{proof}

\begin{proof}[Proof of Lemma~\ref{lem:Kraus-composition-dq}]
    We simply have
    \begin{equation}
        K_2 \circ K_1 = (K_{11}\otimes K_{21}) \circ (K_{12}
        \otimes K_{22}) = (K_{11} \circ K_{12}) \otimes (K_{21}
        \circ K_{22}), 
    \end{equation}
    which is a double-qubit Kraus model in the form of \eqref{eq:Kraus_dq}. 
\end{proof}

\begin{proof}[Proof of Lemma~\ref{lem:prob-to-Kraus}]

The $X, Y, Z$ errors work on a density matrix
\begin{equation*}
    \rho =
    \begin{pmatrix}
        r_{11} & r_{12} \\
        r_{21} & r_{22}
    \end{pmatrix}
\end{equation*}
as
\begin{equation*}
    X \rho X =
    \begin{pmatrix}
        r_{22} & r_{21} \\
        r_{12} & r_{11}
    \end{pmatrix},
    Y \rho Y =
    \begin{pmatrix}
        r_{22} & -r_{21} \\
        -r_{12} & r_{11}
    \end{pmatrix},
    Z \rho Z =
    \begin{pmatrix}
        r_{11} & -r_{12} \\
        -r_{21} & r_{22}
    \end{pmatrix}. 
\end{equation*}
Thus, the density matrix with errors is
\begin{equation*}
    \begin{split}
        \tilde{\rho} = & p_I \rho + p_X X \rho X + p_Y Y \rho Y
        + p_Z Z \rho Z + p_{R0} \ketbra{0}{0} + p_{R1} \ketbra{1}{1}
        + \frac{p_DI}{2}\\
        = &
        \begin{pmatrix}
            (p_I + p_Z) r_{11} + (p_X + p_Y) r_{22} + p_{R0} + \frac{p_D}{2}
            & (p_X - p_Y) r_{21} + (p_I - p_Z) r_{12} \\
            (p_X - p_Y) r_{12} + (p_I - p_Z) r_{21}
            & (p_I + p_Z) r_{22} + (p_X + p_Y) r_{11} + p_{R1} + \frac{p_D}{2}
        \end{pmatrix} \\
        = &
        \begin{pmatrix}
            (p_I + p_Z - p_X - p_Y) r_{11} + p_X + p_Y + p_{R0}
            + \frac{p_D}{2}
            & (p_I - p_Z) r_{12} \\
            (p_I - p_Z) r_{21}
            & \frac{p_D}{2} - (p_I + p_Z - p_X - p_Y) r_{11}
            + p_I + p_Z + p_{R1}
        \end{pmatrix}. 
    \end{split}
\end{equation*}
The last equality makes use of $p_X = p_Y$ and $r_{11} + r_{22} =
\tr{\rho}  =1$. On the other hand, the Kraus error \eqref{eq:Kraus_sq}
works as
\begin{equation*}
    K_\text{sq} (\rho) =
    \begin{pmatrix}
        (a_1^2 + a_2^2 + b_1^2 + b_2^2 - 1) r_{11} + 1 - b_1^2 - b_2^2
        & (a_1 b_1 + a_2 b_2) r_{12} \\
        (a_1 b_1 + a_2 b_2) r_{21}
        & -(a_1^2 + a_2^2 + b_1^2 + b_2^2 - 1) r_{11} + b_1^2 + b_2^2
    \end{pmatrix}.
\end{equation*}
Therefore, as long as it holds
\begin{equation}
    \label{eq:compose-linear}
    \begin{cases}
        a_1^2 + a_2^2 = p_I + p_Z + p_{R0} + \frac{p_D}{2} \\
        b_1^2 + b_2^2 = p_I + p_Z + p_{R1} + \frac{p_D}{2} \\
        a_1 b_1 + a_2 b_2 = p_I - p_Z 
    \end{cases},
\end{equation}
the probabilistic error is equivalent to the Kraus error. By
construction, we have
\begin{equation*}
    \begin{split}
        0 \leq p_I + p_Z + p_{R0} + \frac{p_D}{2} \leq 1, \\
        0 \leq p_I + p_Z + p_{R1} + \frac{p_D}{2} \leq 1. \\
    \end{split}
\end{equation*}
In addition, we have,
\begin{equation*}
    \bracket{p_I + p_Z + p_{R0} + \frac{p_D}{2}}
    \bracket{p_I + p_Z + p_{R1} + \frac{p_D}{2}}
    \ge p_I^2 + p_Z^2 \ge (p_I-p_Z)^2. 
\end{equation*}
Finally, by Lemma~\ref{lem:system-of-eq-Kraus}, we know that
\eqref{eq:compose-linear} has solutions and probabilistic error can be in
the form of \eqref{eq:Kraus_sq_matrix_form}.

\end{proof}

\section{Linear growing error bound} \label{app:linear-error-bound}

\begin{proof}[Proof of Lemma~\ref{lem:linear-bound-for-Kraus-error}]

Suppose the corresponding unitary matrix of gate $G_k$ is $U_k$, then
error-free density matrices are shown in \eqref{eq:error_free}, while the
actual states can also be written explicitly as
\begin{equation*}
    \begin{split}
        \tilde{\rho}_0 & = \rho_0, \\
        \tilde{\rho}_1 & = K_1\bracket{U_1 \tilde{\rho}_0 U_1^\dagger}, \\
        & \cdots \\
        \tilde{\rho}_m & = K_m\bracket{U_m \tilde{\rho}_{m-1} U_m^\dagger}.
    \end{split}
\end{equation*}

By the assumption \eqref{eq:assume-for-Kraus-linear-bound}, we have
\begin{equation} \label{eq:Kraus_error_one_step}
    \begin{split}
        & \fnorm{\tilde{\rho}_k - \rho_k}
        = \fnorm{K_k \bracket{U_k \tilde{\rho}_{k-1} U_k^\dagger}
        - U_k \rho_{k-1} U_k^\dagger} \\
        = & \fnorm{K_k \bracket{U_k \tilde{\rho}_{k-1} U_k^\dagger}
        - U_k \tilde{\rho}_{k-1} U_k^\dagger
        + U_k \tilde{\rho}_{k-1} U_k^\dagger
        - U_k \rho_{k-1} U_k^\dagger} \\
        \le & \fnorm{K_k \bracket{U_k \tilde{\rho}_{k-1} U_k^\dagger}
        - U_k \tilde{\rho}_{k-1} U_k^\dagger}
        + \fnorm{U_k \tilde{\rho}_{k-1} U_k^\dagger
        - U_k \rho_{k-1} U_k^\dagger} \\
        \le & \gamma_1 + \fnorm{\tilde{\rho}_{k-1} - \rho_{k-1}}.
    \end{split}
\end{equation}

Recursively applying \eqref{eq:Kraus_error_one_step}, we obtain,
\begin{equation*}
    \fnorm{\tilde{\rho}_m - \rho_m}
    \le \gamma_1 + \fnorm{\tilde{\rho}_{m-1} - \rho_{m-1}}
    \le 2 \gamma_1 + \fnorm{\tilde{\rho}_{m-2} - \rho_{m-2}}
    \le \cdots
    \le m \gamma_1.
\end{equation*}

\end{proof}

\begin{proof}

For general probabilistic model \eqref{eq:prob-to-Kraus}, it holds
\begin{equation*}
    \begin{split}
        \fnorm{P(\rho)-\rho}
        & = \fnorm{p \rho + (1-p) R \rho R^\dagger - \rho} \\
        & \le (1-p)\bracket{\fnorm{\rho} + \fnorm{R \rho R^\dagger}} \\
        &\le 2(1-p). 
    \end{split}
\end{equation*}
Choosing $\gamma_2 = 2 - 2\min\{p_1,\cdots,p_m\}$ and applying similar
derivations as in the proof of
Lemma~\ref{lem:linear-bound-for-Kraus-error}, we have
\eqref{eq:linear-bound-for-prob}.

\end{proof}

\section{Supporting Inequalities} \label{app:ineq}

We prove a few inequalities in this section, which are widely used
throughout this paper.

Given two density matrices, $\tilde{\rho}$ and $\rho$, we have,
\begin{equation*}
    \fnorm{\tilde{\rho} - \rho}^2
    = \fnorm{\tilde{\rho}}^2 + \fnorm{\rho}^2 - 2 \tr{\tilde{\rho}\rho}
    \leq 1 + 1 - 2 \tr{\tilde{\rho}\rho}
    \leq 2,
\end{equation*}
where the first inequality is due to the property of density matrix, and
the second inequality is due to Ruhe’s trace inequality.

Given two density matrices, $\tilde{\rho}$ and $\rho$, and a Kraus
operator $K$, we have,
\begin{equation} \label{eq:ineq2}
    \fnorm{K(\tilde{\rho}) - \rho}^2
    = \fnorm{K(\tilde{\rho})}^2 + \fnorm{\rho}^2
    - 2 \tr{K(\tilde{\rho})\rho}
    \leq 1 + 1 - 2 \tr{K(\tilde{\rho})\rho}
    \leq 2,
\end{equation}
where the first inequality is partially due to the property of density
matrix, and the second inequality is due to Ruhe’s trace inequality. Since
the Kraus operator preserves the trace and the semi-positivity of the
matrix, we could show that $\fnorm{K(\tilde{\rho})}^2 \leq 1$, which is
used in the first inequality in \eqref{eq:ineq2}.

\section{Kraus error model lemmas} \label{app:Kraus-lemma}

For both Lemma~\ref{lem:Kraus-lemma-sq} and
Lemma~\ref{lem:Kraus-lemma-dq}, we would focus on the proof for density
matrices of pure state and then adopt the inequality
\eqref{eq:Kraus-lemma-decomp} to achieve the final inequalities. Given a
density matrix $\rho$, we could always rewrite it as an eigenvalue
decomposition,
\begin{equation*}
    \rho = \sum_{i = 1}^{N} \lambda_i u_i u_i^\dagger,
\end{equation*}
where $u_i$s are orthonormal vectors and $\lambda_i$s are all
non-negative. Further, we know that the density matrix is of trace one,
i.e., $\tr{\rho} = 1$, which is equivalent to $\sum_{i} \lambda_i = 1$. By
the linearity of Kraus operator and the triangle inequality of Frobenius
norm, we have,
\begin{equation} \label{eq:Kraus-lemma-decomp}
    \fnorm{K(\rho)} = \fnorm{\sum_{i=1}^N \lambda_i K(u_i u_i^\dagger) }
    \le \sum_{i=1}^N \lambda_i \fnorm{K(u_i u_i^\dagger)}
    \le \max_i \fnorm{K(u_i u_i^\dagger)}.
\end{equation}
Hence, it is sufficient to show that Lemma~\ref{lem:Kraus-lemma-sq} and
Lemma~\ref{lem:Kraus-lemma-dq} holds for the density matrices of pure
state. 

Firstly, we prove a lemma for $n$-qubit pure state density matrixs. 

\begin{lemma}
    \label{lemma:app}
    We consider an $n$-qubit density matrix of a pure state 
    \begin{equation} \label{eq:n-qubit-density-matrix}
        \rho =
        \begin{pmatrix}
            r_{11} & r_{12} & \cdots & r_{1N} \\
            r_{21} & r_{22} & \cdots & r_{2N} \\
            \vdots & \vdots & \ddots & \vdots \\
            r_{N1} & r_{N2} & \cdots & r_{NN}
        \end{pmatrix}
        =
        \begin{pmatrix}
            R_{11} & R_{12} \\
            R_{21} & R_{22}
        \end{pmatrix}
        =
        \begin{pmatrix}
            R_1 R_1^\dagger & R_1 R_2^\dagger \\
            R_2 R_1^\dagger & R_2 R_2^\dagger
        \end{pmatrix},
    \end{equation}
    where $R_{11}, R_{12}, R_{21}, R_{22} \in \bbC^{\frac{N}{2} \times
    \frac{N}{2}}$, and $R_1, R_2 \in \bbC^{\frac{N}{2}}$ form the state of
    the pure state. The matrix
    \begin{equation}
    \rho^\prime = 
        \begin{pmatrix}
        \fnorm{R_{11}} & \fnorm{R_{12}} \\
        \fnorm{R_{21}} & \fnorm{R_{22}}
    \end{pmatrix}
    \end{equation}
    is a single-qubit density matrix of a pure state. 
\end{lemma}

\begin{proof}

Both $R_{11} = R_1 R_1^\dagger$ and $R_{22} = R_2 R_2^\dagger$ are
symmetric positive semi-definite matrices and of rank no greater than $1$.
Therefore, it holds
\begin{equation*}
    \tr{\rho^\prime} = \fnorm{R_{11}} + \fnorm{R_{22}}
    = \tr{R_{11}} + \tr{R_{22}} = \tr{\rho} = 1. 
\end{equation*}
We also have
\begin{equation*}
    \fnorm{\rho^\prime}^2 = \fnorm{R_{11}}^2 + \fnorm{R_{12}}^2
    + \fnorm{R_{21}}^2 + \fnorm{R_{22}}^2 = \fnorm{\rho}^2 = 1. 
\end{equation*}
By the construction of $R_{12}$ and $R_{21}$, we know that $\rho^\prime$
is a symmetric matrix. Let the two real eigenvalues of $\rho^\prime$ be
$\lambda_1$ and $\lambda_2$ ($\lambda_1 \geq \lambda_2$). The above
equations are equivalent to,
\begin{equation*}
    \lambda_1 + \lambda_2 = 1, \quad
    \lambda_1^2 + \lambda_2^2 = 1,
\end{equation*}
whose solution is $\lambda_1 = 1$ and $\lambda_2 = 0$. Hence,
$\rho^\prime$ is a single-qubit density matrix of a pure state. 

\end{proof}

Now we prove Lemma~\ref{lem:Kraus-lemma-sq}. 

\begin{proof}[Proof of Lemma~\ref{lem:Kraus-lemma-sq}]

We first prove that Lemma~\ref{lem:Kraus-lemma-sq} holds for single-qubit
density matrix. Denote $\rho = uu^\dagger$, where $u=[u_1,u_2]^\top
\in\mathbb{C}^2$ and $\norm{u}^2=1$. Substituting the expression of
$\Ksq{u u^\dagger}$ as in \eqref{eq:Kraus_sq},
\begin{equation*}
    \begin{split}
        \Ksq{u u^\dagger} & = V_1 u u^\dagger V_1^\dagger
        + V_2 u u^\dagger V_2^\dagger + V_3 u u^\dagger V_3^\dagger
        + V_4 u u^\dagger V_4^\dagger \\
        & = x_1 x_1^\dagger + x_2 x_2^\dagger
        + x_3 x_3^\dagger + x_4 x_4^\dagger,
    \end{split}
\end{equation*}
where $x_k = V_k u$ for $k = 1, 2, 3, 4$, i.e.
\begin{equation*}
    x_1 =
    \begin{pmatrix}
        a_1 u_1 \\
        b_1 u_2
    \end{pmatrix},
    \quad
    x_2 =
    \begin{pmatrix}
        a_2 u_1 \\
        b_2 u_2
    \end{pmatrix},
    \quad
    x_3 =
    \begin{pmatrix}
        0 \\
        a_3 u_1
    \end{pmatrix},
    \quad
    x_4 =
    \begin{pmatrix}
        b_3 u_2 \\
        0
    \end{pmatrix},
\end{equation*}
satisfying
\begin{equation*}
    \norm{x_1}^2 + \norm{x_2}^2 + \norm{x_3}^2 + \norm{x_4}^2
    = u^\dagger \bracket{V_1^\dagger V_1 + V_2^\dagger V_2
    + V_3^\dagger V_3 + V_4^\dagger V_4} u = 1.
\end{equation*}
By the triangle inequality, we have
\begin{equation} \label{eq:delta_ineq_sq}
    \fnorm{\Ksq{u u^\dagger}} \le \fnorm{x_1x_1^\dagger}
    + \fnorm{x_2x_2^\dagger} + \fnorm{x_3x_3^\dagger}
    + \fnorm{x_4x_4^\dagger} = 1.
\end{equation}

Next, we would like to show that the equality in the inequality
\eqref{eq:delta_ineq_sq} cannot be achieved, i.e. $\fnorm{\Ksq{u
u^\dagger}}$ is strictly less than $1$ for any $u$. Let us consider the
above triangle inequality for a pair $x_i$ and $x_j$. The triangle
inequality could be simplified as the Cauchy-Schwarz inequality,
\begin{equation} \label{eq:delta_ineq_sq_2}
    \fnorm{x_i x_i^\dagger + x_j x_j^\dagger} \leq \fnorm{x_i x_i^\dagger} 
    + \fnorm{x_j x_j^\dagger}
    \Leftrightarrow
    \abs{x_i^\dagger x_j} \leq \abs{x_i}^2 \abs{x_j}^2
\end{equation}
where the equality is achieved if and only if $x_i$ and $x_j$ are linearly
dependent. Further, the equality in \eqref{eq:delta_ineq_sq} holds if and
only if the equality in \eqref{eq:delta_ineq_sq_2} holds for any pair
$x_i$ and $x_j$.

Recall that $\det{\begin{pmatrix} a_1 & a_2 \\ b_1 & b_2 \end{pmatrix}}
\neq 0$ as in \eqref{eq:Kraus_sq}. When $u_1 = 0$, by the normality of
$u$, we know that $\abs{u_2} = 1$. $x_1$, $x_2$ and $x_4$ cannot be
linearly dependent unless $b_1 = b_2 = b_3 = 0$, which violates the
assumption. When $u_2 = 0$, similar analysis leads to $a_1 = a_2 = a_3 =
0$, which also violates the assumption. When $u_1 \ne 0$ and $u_2 \neq 0$,
by the linear dependency of $x_3$ and $x_4$, we know that $a_3 = b_3 = 0$.
In this case, the linear dependency of $x_1$ and $x_2$
contradicts the nonzero determinant assumption. Therefore, $\fnorm{\Ksq{u
u^\dagger}}$ is strictly less than $1$.

Since $\fnorm{\Ksq{u u^\dagger}}$ is a continuous function of $u_1$ and
$u_2$, which are defined on a compact domain $\abs{u_1}^2 + \abs{u_2}^2 =
1$, the supremum of $\fnorm{\Ksq{u u^\dagger}}$ is achievable and is
strictly less than $1$. Thus, there is a $\delta^\prime > 0$, such that
$\fnorm{\Ksq{u u^\dagger}}^2 \le 1 - \delta^\prime$ for pure state
single-qubit density matrix $\rho = u u^\dagger$. Via the inequality
\eqref{eq:Kraus-lemma-decomp}, we know that there exists a $\delta > 0$,
such that $\Ksq{\rho} \le 1 - \delta$ for all single-qubit density matrix
$\rho$.

Now we consider an $n$-qubit pure state density matrix in the form of
\eqref{eq:n-qubit-density-matrix}. Without loss of generality, we assume
that the Kraus operator is acted on the first qubit, then
\begin{equation*}
    \begin{split}
        \Ksq{\rho} & = \sum_{j=1}^4 \bracket{V_j \otimes I}
        \rho \bracket{V_j^\dagger \otimes I} \\
        & =
        \begin{pmatrix}
            b_3^2 R_{22} + \bracket{a_1^2 + a_2^2} R_{11}
            & \bracket{a_1 b_1 + a_2 b_2} R_{12} \\
            \bracket{a_1 b_1 + a_2 b_2} R_{21} & a_3^2 R_{11}
            + \bracket{b_1^2 + b_2^2} R_{22}
        \end{pmatrix}.
    \end{split}
\end{equation*}
Using triangle inequality, we have
\begin{equation*}
    \begin{split}
        \fnorm{\Ksq{\rho}}^2 = & \fnorm{b_3^2 R_{22} 
        + \bracket{a_1^2 + a_2^2} R_{11}}^2
        + \fnorm{\bracket{a_1 b_1 + a_2 b_2} R_{12}}^2 \\
        & + \fnorm{a_3^2 R_{11} + \bracket{b_1^2 + b_2^2} R_{22}}^2
        + \fnorm{\bracket{a_1 b_1 + a_2 b_2} R_{21}}^2 \\
        \le & \bracket{b_3^2 \fnorm{R_{22}} + \bracket{a_1^2 + a_2^2}
        \fnorm{R_{11}}}^2 + \bracket{a_1 b_1 + a_2 b_2}^2 \fnorm{R_{12}}^2 \\
        & + \bracket{a_3^2 \fnorm{R_{11}} + \bracket{b_1^2 + b_2^2}
        \fnorm{R_{22}}}^2 + \bracket{a_1 b_1 + a_2 b_2}^2 \fnorm{R_{21}}^2 \\
        = & \fnorm{\Ksq{\rho^\prime}}^2, 
    \end{split}
\end{equation*}
where
\begin{equation*}
    \rho^\prime =
    \begin{pmatrix}
        \fnorm{R_{11}} & \fnorm{R_{12}} \\
        \fnorm{R_{21}} & \fnorm{R_{22}}
    \end{pmatrix},
\end{equation*}
and we abuse notation $\Ksq{\rho^\prime}$ denoting $K_\text{sq}$ acting on
$\rho^\prime$. According to Lemma~\ref{lemma:app}, $\rho^\prime$ is a
single-qubit pure state density matrix. Therefore, it holds
$\fnorm{K_\text{sq}\bracket{\rho}}^2 \le
\fnorm{K_\text{sq}\bracket{\rho^\prime}}^2 \le 1-\delta^\prime$ for some
constant $\delta^\prime$. Again, via the inequality
\eqref{eq:Kraus-lemma-decomp}, we know that there exists a $\delta > 0$,
such that $\Ksq{\rho} \le 1 - \delta$ for all density matrix $\rho$.

\end{proof}

Lemma~\ref{lem:Kraus-lemma-dq} can be proved in the same way as that of
Lemma~\ref{lem:Kraus-composition-sq}. For the sake of space, we omit the
detailed proof.

\end{document}